\documentclass[11pt,leqno]{article}
\usepackage{amsthm,amsmath, mathtools}
\usepackage[round]{natbib}
\usepackage{amsfonts}
\usepackage{ifthen}
\usepackage{enumitem}
\usepackage{amssymb,dsfont,url,color,booktabs}
 \usepackage{tikz}
 \usepackage{setspace}

\usepackage{epstopdf}
\usepackage{chngcntr}
\usepackage{bbm}
\usepackage{array}
\usepackage[export]{adjustbox}[2011/08/13] 
\usepackage{geometry}
\usepackage{rotating}
\usepackage{todonotes}
\usepackage{xspace}
\usepackage{algorithm}
\usepackage[noend]{algpseudocode}
 \usepackage[english]{babel}
\usepackage{bbm}
\usepackage{float}
\usepackage{caption}
\usepackage{multirow}
\usepackage{booktabs}
\usepackage[font={footnotesize}]{caption}
\usepackage{placeins}
\usepackage{subfig} 		

\theoremstyle{plain}
\newtheorem{theorem}{Theorem}[section]

\newtheorem{proposition}[theorem]{Proposition}

\newcommand{\notiz}[1]{\relax}

\newcommand{\zitep}[1]{\relax}

\newcommand{\1}{\mathds 1}            

\newcommand{\Price}[1][]{
		\ifthenelse{\equal{#1}{}}{\mathit{Price}}{\Price{}^{#1}}
	} 

\newlength{\wordlength}

\newcommand{\ul}{\underline}
\newcommand{\ol}{\overline}
\newcommand{\RR}{\mathbb{R}}
\newcommand{\EE}{\mathbb{E}}

\newcommand{\dy}{\text{d}y}

\renewcommand{\cite}{\citet}

\numberwithin{equation}{section}
\numberwithin{figure}{section}
\numberwithin{table}{section}

\begin{document}
\title{Fast Calculation of Credit Exposures for Barrier and Bermudan options using Chebyshev interpolation}

\bigskip
\author{\textbf{Kathrin Glau$\vphantom{l}^{1,2}$,} \textbf{Ricardo Pachon$\vphantom{l}^{3}$,} \textbf{Christian P{\"o}tz$\vphantom{l}^{1}$
}
\\\\$\vphantom{l}^{\text{1}}$Queen Mary University of London, UK\\
$\vphantom{l}^{\text{2}}$Ecole polytechnique f\'ed\'erale de Lausanne, Switzerland\\
$\vphantom{l}^{\text{3}}$Credit Suisse, UK
}

\maketitle
\begin{abstract}
We introduce a new method to calculate the credit exposure of Bermudan, discretely monitored barrier and European options. Core of the approach is the application of the dynamic Chebyshev method of \cite{GlauMahlstedtPoetz2019}. The dynamic Chebyshev method delivers a closed form approximation of the option prices along the paths together with the options' delta and gamma. Key advantage is the polynomial structure of the approximation, which allows us a highly efficient evaluation of the credit exposures, even for a large number of simulated paths. The approach is highly flexible in the model choice, payoff profiles and asset classes.  We compute the exposure profiles for Bermudan and barrier options in three different equity models and compare them to the profiles of European options. The analysis reveals potential shortcomings of common simplifications in the exposure calculation. The proposed method is sufficiently simple and efficient to avoid such risk-bearing simplifications.

\end{abstract}

\textbf{Keywords}
	Option Pricing, Credit Exposure, Complexity Reduction, Polynomial Interpolation
	
\noindent\textbf{2010 MSC} 91G60, 41A10  

\section{Introduction}
The credit exposure resulting from two counterparties facing each other on a derivatives deal is the main input in a growing list of calculations, all crucial since the financial crisis of 2007--2008. Credit exposures are used to estimate, for example, counterparty credit risk (and consequently the regulatory capital of financial firms), initial margins of collateralized trades, Credit Valuation Adjustments (CVA), Debit Valuation Adjustments (DVA) and, more recently, Funding Valuation Adjustments (FVA).

The exposure of a trade at time $t$ is defined as
\[
E_t(X_t) =\max\{V_t(X_t),0\},
\]
where $X_t$ is the risk factor that drives the price $V_t$ at time $t$ of a portfolio of derivatives. In essence, the credit exposure calculation projects forward in time the distributions of relevant underlying assets, which follow appropriate stochastic models, and obtains the associated distributions of the values of the derivatives in scope, up to their longest maturity. The specifics of this calculation vary with each application. For example, for CVA and DVA the calculation is performed at netting set  while for FVA is done at portfolio level. For CVA, negative exposures are floored to zero before taking a discounted average, i.e., the expected future exposure of a trade used in CVA is calculated as
\begin{align}\label{eq:EE}
EE_t(X_t) = \mathbb{E}^{\mathbb{P}} [E_t(X_t)]
\end{align}
while for counterparty risk the potential future exposure (PFE) at a percentile of the exposure (typically 95th percentile) is usually calculated, i.e.,  for a given level $\alpha \in (0,1)$ it is defined as
\begin{align}\label{eq:PFE}
PFE_t^\alpha (X_t) =\inf \{y: \mathbb{P}(E_t(X_t)\leq y)\geq \alpha\}.
\end{align}

The mentioned distributions are usually obtained through Monte Carlo simulation: On some chosen time points, the derivatives are re-evaluated on various scenarios, randomly drawn from the distribution of the underlying asset, and from the resulting distribution the required metric is extracted. The crux of the calculation is the repeatedly call of the pricers which, for exotic trades, can be computationally expensive. See Gregory (2010) for an overview of credit exposure and its calculation. 

In this paper we introduce a numerical technique based on Chebyshev interpolation for the fast calculation of credit exposures of barrier and Bermudan options, which due to their path-dependant nature make their computation expensive. Especially, when using the simple but naive approach of calling Monte Carlo simulations within a Monte Carlo simulation. In the literature regression based methods are studied in order to avoid nested Monte Carlo simulation, see for instance \cite{Schoeftner2008}, who calculate the exposure and CVA for derivatives without analytic solution (e.g. Bermudan options) based on a modification of the Least-Squares Monte Carlo approach of \cite{LongstaffSchwartz2001}. Another approach is investigated in \cite{ShenWeideAnderluh2013}, who calculate the exposure for Bermudan options one one asset, based on the COS method for early-exercise options of \cite{FangOosterlee2009}.

The study of Chebyshev interpolation belongs to the field of Approximation Theory, a well established branch of mathematics, and the results that provide a framework for its efficiency span over a hundred years. Being a tool with such a long history we anticipate the reader to be aware of it, however we believe that its many advantages for practical applications have been overlooked until recently, even in the specialised community.  See \cite{Trefethen2013} for an overview on Chebyshev interpolation and Approximation Theory.

In the computational finance literature we have not been able to locate more than handful of references that exploit the beneficial properties of Chebyshev interpolation or Chebyshev series. In this paper we aim to help closing this gap. We argue that Chebyshev interpolation can assist in the calculation of credit exposures for path-dependant products by producing an approximation to the pricer with some outstanding qualities: it can be quickly constructed from a few evaluations on a (non-adaptive) grid of asset values; it is robust and efficient to evaluate; and its accuracy can be tuned even for high orders.

The structure of this paper is as follows. In section 2 we present the basics of Chebyshev interpolation. Section 3 introduces the numerical technique for the calculation of exposures based on the dynamic Chebyshev algorithm of \cite{GlauMahlstedtPoetz2019} and  numerical experiments are presented in section 4. In section 5 we discuss the behaviour of the credit exposure profiles for barrier and Bermudan options and close the paper with conclusion in section 6.\\

\section{Chebyshev interpolation}
The one-dimensional Chebyshev interpolation is a polynomial interpolation of a function $f$ in the interval $[-1,1]$ of degree $N$ in the $N+1$ Chebyshev points $z_{k}=\cos(\pi k/N)$. These points are not equidistantly distributed but cluster at $-1$ and $1$. The interpolant can be written as a sum Chebyshev polynomials $T_{j}(z)=cos(j\,\text{acos}(z))$ with an explicit formulas for the coefficients, i.e. for a function $f:[-1,1]\rightarrow\RR$ we obtain
\begin{align*}
I_{N}(f)(z)=\sum_{j=0}^{N}c_{j}T_{j}(z) \quad \text{with} \quad c_{j}=\frac{2^{\1_{\{0<j<N\}}}}{N}\sum_{k=0}^{N}{}^{''}f(z_{k})T_{j}(z_{k})
\end{align*}
where $\sum{}^{''}$ indicates the summand is multiplied by $1/2$ if $k=0$ or $k=N$. In order to evaluate the interpolation efficiently one can exploit the following alternative definition of the Chebyshev polynomials
\begin{align}\label{Chebpoly_recur}
T_{n+1}(z)=2zT_{n}(z)-T_{n-1}(z), \qquad T_{1}(z)=z \quad \text{and} \quad T_{0}(z)=1.
\end{align}
Based on this recurrence relation Clenshaw's algorithm provides an efficient framework to evaluate the Chebyshev interpolant $I_{N}(f)$
\begin{align*}
&b_{k}(x)=c_{k}+ 2xb_{k+1}(x)-b_{k+2}(x),\quad\text{for}\quad k=n,\ldots,1\\
&I_{N}(f)(x)=c_{0}+xb_{1}(x)-b_{2}(x)
\end{align*}
with starting values $b_{n+1}(x)=b_{n+2}(x)=0$.

In order to interpolate functions on an arbitrary rectangular $\mathcal{X}=[\underline{x},\overline{x}]$, we introduce a transformation $\tau_{\mathcal{X}}:[-1,1]\rightarrow\mathcal{X}$ defined by
\begin{align}
\tau_{\mathcal{X}}(z)=\overline{x}+0.5(\underline{x}-\overline{x})(1-z).\label{Transformation}
\end{align}
The Chebyshev interpolation of a function $f:\mathcal{X}\rightarrow\mathbb{R}$ can be written as
\begin{align}\label{Cheby_Interpolation}
I_{\overline{N}}(f)(x)=\sum_{j=0}^{N}c_{j}p_{j}(x) \quad \text{with}\quad c_j&=\frac{2^{\1_{\{0<j<N\}}}}{N_i}\sum_{k=0}^{N}{}^{''}f(x_k)T_{j}(z_k)
\end{align}
for $x\in\mathcal{X}$ with transformed Chebyshev polynomials $p_j(x)=T_j(\tau^{-1}_{\mathcal{X}}(x))1_{\mathcal{X}}(x)$ and transformed Chebyshev points $x_k=\tau_{\mathcal{X}}(z_k)$.
The one-dimensional interpolation has a tensor based extension to the multivariate case, see e.g. \cite{SauterSchwab2010}.\\

The Chebyshev interpolation provides promising convergence results and explicit error bounds. The interpolation converges for all Lipschitz continuous functions and for analytic functions the interpolation converges exponentially fast. See \cite{Trefethen2013} for the one-dimensional case and for a multivariate version \cite{SauterSchwab2010}. Moreover, the convergence is of polynomial order for differentiable functions and the derivatives converge as well, see \cite{GassGlauMahlstedtMair2018}. 

The Chebyshev interpolation is implemented in the open-source $\textit{MatLab}$ package $\textit{chebfun}$ available at $\textit{www.chebfun.org}$. We use this package in some of our numerical experiments. 

\section{A unified approach for exposure calculation}
In this section we investigate the exposure calculation for different types of options such as European, Bermudan and barrier options. We propose a unified approach for all three types of options based on the Dynamic Chebyshev algorithm of \cite{GlauMahlstedtPoetz2019}. The core idea is to write the option price as a solution of a Dynamic Programming problem and to approximate the solution with Chebyshev polynomials.\\

\subsection{The Dynamic Chebyshev approach for exposure calculation}
For many (portfolios of) derivatives the expected exposure as defined in \eqref{eq:EE} cannot be calculated analytically and simulation approaches come into play. The risk factors $X_{t}^{i}$, $i=1,\ldots,M$ are simulated and the expected exposure is approximated by
\begin{align*}
EE_{t}(x)=\EE^{\mathbb{P}}[\max\{V_{t}(X_{t}),0\}]\approx\frac{1}{M}\sum_{i=1}^{M}\max\{V_{t}(X_{t}^{i}),0\}.
\end{align*}
Hence, the values $V_{t}(X_{t}^{i})$ of the derivative have to be calculated for a large number $M$ of simulated risk factors. Typically, there is no analytic solution available and the evaluation becomes computationally demanding. This is especially the case when the value function $V_{t}$ at time point $t$ depends on the conditional expectation of the value function at $t+1$. 

In order to address this issue we propose to approximate the function $x\mapsto V_{t}(x)$ with a polynomial. More precisely, we will approximate the value function with a weighted sum of Chebyshev polynomials, i.e.
\begin{align*}
V_{t}(x)\approx\widehat{V}_{t}(x)=\sum_{j=0}^{N}c_{j}p_{j}(x)
\end{align*}
with weights/coefficients $c_{j}$. Then we replace the value function with its Chebyshev approximation in the exposure calculation
\begin{align}\label{eq:EE_calc_Cheb_approx}
EE_{t}(x)=\EE^{\mathbb{P}}[\max\{V_{t}(X_{t}),0\}]\approx\frac{1}{M}\sum_{i=1}^{M}\max\{\widehat{V}_{t}(X_{t}^{i}),0\}.
\end{align}
Even for a large number of simulated risk factors the sum of polynomials can be evaluated efficiently. The remaining question is how to obtain the coefficients $c_{j}$ in every time step. Fortunately, when using the Chebyshev interpolation for approximation we have an explicit formula for the coefficients, which is a linear transformation of the function values at the Chebyshev points. The crucial point is thus the efficient calculation of the function values at the set of nodal points. Here, the Dynamic Chebyshev algorithm comes into play. The Dynamic Chebyshev method was presented in \cite{GlauMahlstedtPoetz2019} as a pricing method and can be very easily extended to calculate expected exposures of options.\\

In order to introduce the algorithm, we start with the pricing of a Bermudan option. The value of a Bermudan option with payoff $g$ and exercise dates $t_{0},\ldots,t_{n}=T$ is given by the optimal stopping problem
\begin{align*}
V_{t_{0}}(x)=\sup_{t_{0}\leq t_{u}\leq T}\EE^{\mathbb{Q}}[g(X_{t_{u}})\vert X_{t_{0}}=x].
\end{align*}
The principle of Dynamic Programming yields the backward induction
\begin{align*}
&V_{T}(x)=g(x)\\
&V_{t_{u}}(x)=\max\left\{g(x),\EE^{\mathbb{Q}}\left[V_{t_{u+1}}(X_{t_{u+1}})\vert X_{t_{u}}=x\right]\right\}.
\end{align*}
More generally, we can write the value function $V_{t_{u}}(x)$ as
\begin{align}\label{DPP_f}
V_{t_{u}}(x)=f\left(g(t_{u},x),\EE^{\mathbb{Q}}\left[V_{t_{u+1}}(X_{t_{u+1}})\vert X_{t_{u}}=x\right]\right)
\end{align}
for a Lipschitz continuous function $f:\RR\times\RR\rightarrow\RR$ and a function $g:[0,T]\times\RR\rightarrow\RR$ with $g(T,x)=g(x)$. This formulation includes also the pricing of European and barrier options, as we will see later.\\

Assume we have at $t_{u+1}$ an approximation $\widehat{V}_{t_{u+1}}$ with $V_{t_{u+1}}(x)\approx\widehat{V}_{t_{u+1}}(x)=\sum_{j}c_{j}(t_{u+1})p_{j}(x)$. In order to interpolate the value function $V_{t_{u}}$ at $t_{u}$ we have to calculate the values at the Chebyshev nodes $x_{k}$, $k=0,\ldots,N$. In this case the backward induction becomes
\begin{align*}
V_{t_{u}}(x_{k})&=f\left(g(t_{u},x_{k}),\EE^{\mathbb{Q}}\left[V_{t_{u+1}}(X_{t_{u+1}})\vert X_{t_{u}}=x_{k}\right]\right)\\
&\approx f\Big(g(t_{u},x_{k}),\EE^{\mathbb{Q}}\Big[\sum_{j=0}^{N}c_{j}(t_{u+1})p_{j}(X_{t_{u+1}})\vert X_{t_{u}}=x_{k}\Big]\Big)\\
&=f\Big(g(t_{u},x_{k}),\sum_{j=0}^{N}c_{j}(t_{u+1})\EE^{\mathbb{Q}}\left[p_{j}(X_{t_{u+1}})\vert X_{t_{u}}=x_{k}\right]\Big),
\end{align*}
where we exploited the linearity of the conditional expectation. Here, we see that the coefficients $c_{j}$ carry the information of the payoff, and the conditional expectations $\EE^{\mathbb{Q}}\left[p_{j}(X_{t_{u+1}})\vert X_{t_{u}}=x_{k}\right]$ carry the information of the stochastic process. Since the conditional expectations are independent of the backward induction they can be pre-computed in an offline step before the actual pricing. Given values $V_{t_{u}}(x_{k})$, $k=0,\ldots,N$ the Chebyshev coefficients are given by
\begin{align*}
c_{j}(t_{u})=\frac{2^{1_{0<j<N}}}{N}\sum_{k=0}^{N}{}^{''}V_{t_{u}}(x_{k})T_{j}(z_{k}),
\end{align*}
and we obtain a closed form approximation of the option price
\begin{align*}
V_{t_{u}}(x)\approx\widehat{V}_{t_{u}}(x)=\sum_{j=0}^{N}c_{j}(t_{u})p_{j}(x).
\end{align*}
The presented procedure is a pricing method for a large class of option pricing problems which can be written in the form of \eqref{DPP_f}. This includes different option types, payoff profiles as well as different asset classes and models. Note that we presented the framework for an option on one underlying. In case of multiple underlyings we only need to replace the one-dimensional Chebyshev interpolation with its multivariate extension. The more general multivariate version of the algorithm is presented in \cite{GlauMahlstedtPoetz2019}.\\

Three examples of options that can be written as Dynamic Programming problem in the form of \eqref{DPP_f} are early-exercise options (Bermudan options), classical European options and barrier options.\\

\textbf{Bermudan options:}\\
In this case the value function is given as
\begin{align*}
f\left(g(x),\EE^{\mathbb{Q}}\left[V_{t_{u+1}}(X_{t_{u+1}})\vert X_{t_{u}}=x\right]\right)=\max\left\{g(x),\EE^{\mathbb{Q}}\left[V_{t_{u+1}}(X_{t_{u+1}})\vert X_{t_{u}}=x\right]\right\}.
\end{align*}

\textbf{European options:}\\
European options correspond to Bermudan options with no early exercise. In this case the value function becomes
\begin{align*}
f\left(g(x),\EE^{\mathbb{Q}}\left[V_{t_{u+1}}(X_{t_{u+1}})\vert X_{t_{u}}=x\right]\right)=\EE^{\mathbb{Q}}\left[V_{t_{u+1}}(X_{t_{u+1}})\vert X_{t_{u}}=x\right].
\end{align*}

\textbf{Barrier options:}\\
Discretely monitored up-and-out barrier option with barrier $B$ can be written in the same form with value function
\begin{align*}
f\left(g(x),\EE^{\mathbb{Q}}\left[V_{t_{u+1}}(X_{t_{u+1}})\vert X_{t_{u}}=x\right]\right)=\EE^{\mathbb{Q}}\left[V_{t_{u+1}}(X_{t_{u+1}})\vert X_{t_{u}}=x\right]\1_{x\leq B}.
\end{align*}
Similarly, we can use the framework for down-and-out barrier options.\\

The resulting pricing algorithm is for all three problems essentially the same. However, the efficiency of the method is directly related to the smoothness of the value function. As a result the number of nodal points required for a given accuracy varies, compare Section 5.2 and 5.3 in \cite{GlauMahlstedtPoetz2019}.\\

Now, we are in a position to efficiently evaluate the exposure in formula \eqref{eq:EE_calc_Cheb_approx}. Assume we have simulated $M$ paths of the underlying risk factor. Then we price the option along the paths using the closed form approximation in terms of Chebyshev polynomials
\begin{align*}
V_{t_{u}}(X_{t_{u}}^{i})\approx\sum_{j=0}^{N}c_{j}(t_{u})p_{j}(X_{t_{u}}^{i}) \quad \text{for} \quad i=1,\ldots,M \ \text{ and } \ u=0,\ldots,n.
\end{align*}
These values can now be used to calculate the expected exposure or the potential future exposure for a given level $\alpha$. In the case of a Bermudan option one has to take into account that by exercising the option at $t_{u}$ the exposure becomes zero. Similarly, if the barrier option is knocked out the exposure at all future time steps is zero. These two effects yield a decreasing exposure for both types of options. The exposure calculation is described in Algorithm 1 for Bermudan options, in Algorithm 2 for European options and in Algorithm 3 for barrier options. Here, we assume that the days used for the exposure calculation are the exercise days of the option. An adoption of the algorithm for other exposure calculation days is straightforward.\\

\textbf{Algorithm 1: Bermudan options}\\
This algorithm provides a framework to calculate the expected exposure and the potential future exposure for a Bermudan option.
\begin{itemize}
\item[1.] Simulate paths $X_{t_{0}}^{i},\ldots,X_{t_{n}}^{i}$, $i=1,\ldots,M$ under the real world measure $\mathbb{P}$.
\item[2.] Find a suitable domain $\mathcal{X}=[\ul{x},\ol{x}]$ and define nodal points $x_{k}$, $k=0,\ldots,N$.
\item[3.] Pre-compute conditional expectations under the pricing measure $\mathbb{Q}$
\[\Gamma_{k,j}=\EE^{\mathbb{Q}}\left[p_{j}(X_{\Delta t})\vert X_{0}=x_{k}\right].\]
\item[4.] Start pricing at $T$: Compute nodal values $\widehat{V}_{T}(x_{k})=g(x_{k})$ for all $k=0,\ldots,N$ and calculate Chebyshev coefficients $c_{j}(T)$. For all paths compute the exposure $E_{T}^{i}=\max\{g(X_{T}^{i}),0\}$.
\item[5.] Iterative time stepping $t_{u+1}\rightarrow t_{u}$: Assume we have a Chebyshev approximation $V_{t_{u+1}}(x)\approx\widehat{V}_{t_{u+1}}(x)=\sum_{j}c_{j}(t_{u+1})p_{j}(x)$
\begin{itemize}
\item compute nodal values $\widehat{V}_{t_{u}}(x_{k})=\max\{g(x_k),\sum_{j=0}^{N}c_{j}(t_{u+1})\Gamma_{k,j}\}$ and new coefficients $c_{j}(t_u)$,
\item price the option for all simulation paths $V_{t_{u}}^{i}=\widehat{V}_{t_u}(X_{t_{u}}^{i})=\sum_{j\in J}c_{j}(t_{u})p_{j}(X_{t_{u}}^{i})$,
\item calculate exposure $E_{t_{u}}^{i}=\max\{V_{t_{u}}^{i},0\}$,
\item if the option is exercised (i.e $V_{t_{u}}^{i}=g(X_{t_{u}}^{i})$), update the exposure at all future time steps on this path $E_{t_{j}}^{i}$, $j=u+1,\ldots,n$.
\end{itemize}
\item[6.] Obtain an approximation of the option price at $t_{0}$ 
\[\widehat{V}_{t_{0}}(x)=\sum_{j\in J}c_{j}(t_{0})p_{j}(x),\]
an approximation of the expected future exposures
\[EE_{t_{u}}=\EE^{\mathbb{P}}\left[\max\{V_{t_u},0\}\right]\approx\frac{1}{M}\sum_{i=1}^{M}E_{t_{u}}^{i}\quad \text{for all} \quad u=0,\ldots,n,\]
and an approximation of the potential future exposures
\[PFE_{t_{u}}^{\alpha}(x)=\inf\big\{y:\mathbb{P}\big(E_{t}(x)\leq y\big)\geq \alpha\big\}\approx\inf\big\{y:\frac{\#\{E^{i}_{t_{u}}\leq y\}}{M}\geq \alpha\big\}.\]
\end{itemize}
We refer to Step 1 as the simulation phase, Step 2 and 3 as the pre-computation phase and Step 4, 5 and 6 as the time-stepping of the method.\\

\textbf{Algorithm 2: European options}\\
Algorithm 1 can be modified to calculate the exposure of European options. In this case we do not exercise the option until maturity and thus only calculate the continuation value. More precisely, Step 5 in the Bermudan option algorithm is replaced by
\begin{itemize}
\item[5.] Iterative time stepping $t_{u+1}\rightarrow t_{u}$: Assume we have a Chebyshev approximation $V_{t_{u+1}}(x)\approx\widehat{V}_{t_{u+1}}(x)=\sum_{j\in J}c_{j}(t_{u+1})p_{j}(x)$,
\begin{itemize}
\item compute nodal values $\widehat{V}_{t_{u}}(x_{k})=\sum_{j\in J}c_{j}\Gamma_{k,j}$ and new coefficients $c_{j}(t_u)$,
\item price the option for all simulation paths $V_{t_{u}}^{i}=\widehat{V}_{t_u}(X_{t_{u}}^{i})=\sum_{j\in J}c_{j}(t_{u})p_{j}(X_{t_{u}}^{i})$,
\item calculate the exposure $E_{t_{u}}^{i}=\max\{V_{t_{u}}^{i},0\}$.
\end{itemize}
\end{itemize}

\textbf{Algorithm 3: Barrier options}\\
Moreover, Algorithm 1 can be modified to calculate the exposure of barrier options. In this case the interpolation domain is chosen depending on the barrier. There is no early exercise, however, we need to take care of the knock-out feature. For an up-and-out option with barrier $b$, Step 2 and Step 5 in the Bermudan option algorithm are replaced by
\begin{itemize}
\item[2.] Find a suitable domain $\mathcal{X}=[\ul{x},b]$ and define nodal points $x_{k}$, $k=0,\ldots,N$.
\item[5.] Iterative time stepping $t_{u+1}\rightarrow t_{u}$: Assume we have a Chebyshev approximation $V_{t_{u+1}}(x)\approx\widehat{V}_{t_{u+1}}(x)=\sum_{j}c_{j}(t_{u+1})p_{j}(x)$,
\begin{itemize}
\item compute nodal values $\widehat{V}_{t_{u}}(x_{k})=\sum_{j=0}^{N}c_{j}\Gamma_{k,j}$ and new coefficients $c_{j}(t_u)$,
\item price the option for all simulation paths $V_{t_{u}}^{i}=\widehat{V}_{t_u}(X_{t_{u}}^{i})=\sum_{j\in J}c_{j}(t_{u})p_{j}(X_{t_{u}}^{i})$ if $X_{t_{u}}^{i}\leq b$ and $V_{t_{u}}^{i}=0$ otherwise,
\item calculate the exposure $E_{t_{u}}^{i}=\max\{V_{t_{u}}^{i},0\}$,
\item if the option is knocked-out, i.e if $X_{t_{u}}^{i}> b$ update the exposure at all future time steps on this path $E_{t_{j}}^{i}$, $j=u+1,\ldots,n$.
\end{itemize}
\end{itemize}

\subsection{Conceptional benefits of the method}
The presented algorithms provide efficient solutions for the exposure calculation. Moreover, the structure of the new approach comes with conceptual benefits, which can be exploited in practice.\\

\textbf{Efficient computation of the conditional moments}\\
The conditional expectations of the Chebyshev polynomials $\EE^{\mathbb{Q}}\left[p_{j}(X_{t_{u+1}})\vert X_{t_{u}}=x_{k}\right]$ depend only on the underlying process and can be pre-computed prior to the time-stepping. Here two different cases have to be distinguished.\\
If the underlying process $X_{t_{u+1}}\vert X_{t_{u}}=x$ is normally distributed the conditional expectations of the Chebyshev polynomials can be calculated analytically. Examples are the Black-Scholes model (with log-stock price $X_t$), the Vasicek model or the one factor Hull-White model (both with interest rate $X_t$). More generaly, assume for instance the underlying process is modelled via an SDE of the form
\begin{align*}
\text{d}X_{t}=\alpha(t,X_{t})\text{d}t + \beta(t,X_{t})\text{d}W_{t}
\end{align*}
for a standard Brownian motion $W_{t}$ with Euler–Maruyama approximation
\begin{align*}
X_{t_{u+1}}\approx x + \alpha(t_u,x)(t_{u+1}-t_{u}) + \beta(t_{u},x)\sqrt{t_{u+1}-t_{u}}Z =:\widehat{X}^{x}_{t_{u+1}} \qquad Z\sim\mathcal{N}(0,1)
\end{align*}
and the right hand side is thus normally distributed. The following proposition provides an analytic formula for the conditional moments $\EE^{\mathbb{Q}}[p_{j}(\widehat{X}^{x_{k}}_{t_{u+1}})]$.
\begin{proposition}\label{prop:moments_BS_model}
Assume that $X_{t}$ is a stochastic process with $X_{t_{u+1}}\vert X_{t_{u}}=x_{k}\sim\mathcal{N}(x_{k}+\Delta t\:\!\:\!\mu,\Delta t\sigma^{2})$ with $\Delta t=t_{u+1}-t_{u}$. Then the conditional moments can be written as
\begin{align*}
&\EE[p_j(X_{t_{u+1}})\vert X_{t_u}=x]=\EE[T_{j}(Y)\1_{[-1,1]}(Y)]\\
&Y\sim\mathcal{N}\Big(1-2\frac{\ol{x}-x}{\ol{x}-\ul{x}}+\frac{2}{\ol{x}-\ul{x}}\Delta t\:\!\:\!\mu, \big(\frac{2}{\ol{x}-\ul{x}}\big)^{2}\Delta t\sigma^{2}\Big).
\end{align*}
\end{proposition}
\begin{proof}
From the properties of a Brownian motion with drift follows
\begin{align*}
\EE[p_j(X_{t_{u+1}})\vert X_{t_u}=x]=\EE[p_j(x+(X_{t_{u+1}}-X_{t_{u}}))]=\EE[p_j(x+X_{\Delta t})].
\end{align*}
The definition of $p_{j}$ and the inverse of the linear transformation $\tau_{[\ul{x},\ol{x}]}$ yield
\begin{align*}
\EE[p_j(x+X_{\Delta t})]&=\EE[T_j(\tau^{-1}_{[\ul{x},\ol{x}]}(x+X_{\Delta t})\1_{[\ul{x},\ol{x}]}(x+X_{\Delta t})]\\
&=\EE[T_{j}(1-2\frac{\ol{x}-(x+X_{\Delta t})}{\ol{x}-\ul{x}})\1_{[\ul{x},\ol{x}]}(x+X_{\Delta t})]\\
&=\EE[T_{j}(1-2\frac{\ol{x}-x}{\ol{x}-\ul{x}}+\frac{2}{\ol{x}-\ul{x}}X_{\Delta t})\1_{[\ul{x},\ol{x}]}(x+X_{\Delta t})]\\
&=\EE[T_{j}(Y)\1_{[-1,1]}(Y)]
\end{align*}
with $Y$ defined as
\begin{align*}
Y=1-2\frac{\ol{x}-x}{\ol{x}-\ul{x}}+\frac{2}{\ol{x}-\ul{x}}X_{\Delta t}
\end{align*}
and we used that for a linear transformation holds
\begin{align*}
\1_{[\ul{x},\ol{x}]}(x+X_{\Delta t})]=\1_{[\tau^{-1}_{[\ul{x},\ol{x}]}(\ul{x}),\tau^{-1}_{[\ul{x},\ol{x}]}(\ol{x})]}(\tau^{-1}_{[\ul{x},\ol{x}]}(x+X_{\Delta t}))=\1_{[-1,1]}(Y).
\end{align*}
The properties of a normally distributed variable yields our claim.
\end{proof}

\begin{proposition}
Let $Y\sim\mathcal{N}(\mu,\sigma^{2})$ be a normally distributed random variable with density $f$ and distribution function $F$. The truncated generalized moments $\mu_{j}=\EE[T_{j}(Y)\1_{[-1,1]}(Y)]$ are recursively defined by
\begin{align*}
\mu_{n+1}=2\mu\mu_{n} - 2\sigma^{2}\big(f(1)-f(-1)T_{n}(-1)-2(n-1)\sum_{j=0}^{n-2}{}^{'}\mu_{j}\1_{(n+j)\bmod 2=0}\big)-\mu_{n-1}
\end{align*}
for $n\geq 1$ and starting values $\mu_{0}=F(1)-F(-1)$, $\mu_{1}=\mu\mu_{0}-\sigma^{2}(f(1)-f(-1)$ and where $\sum{}^{'}$ indicates that the first term is multiplied with $1/2$.
\end{proposition}
\begin{proof}
The proof can be found in the appendix.
\end{proof}
For a large model class for which the underlying process is conditionally normally distributed or can be approximated by such a process, the conditional moments can thus be efficiently computed by an analytic formula.\\

If the underlying process is not normally distributed numerical approximation techniques come into play. \cite{GlauMahlstedtPoetz2019} give an overview of different approaches which can be used to calculate the conditional expectations. For example numerical quadrature techniques using the density or characteristic function of the process or with the help of Monte Carlo simulations. The possibility to use different approaches gives us the flexibility to apply the method in a variety of models.\\

When we use an equidistant time stepping $t_{u+1}-t_{u}=\Delta t$ the problem can be further simplified. Assuming
\begin{align}\label{eq_stationarity_assumption}
\EE^{\mathbb{Q}}\left[p_{j}(X_{t_{u+1}})\vert X_{t_{u}}=x_{k}\right]
=\EE^{\mathbb{Q}}\left[p_{j}(X_{\Delta t})\vert X_{0}=x_{k}\right],
\end{align}
the pre-computation step becomes independent of the maturity $T$ and the number of time steps $n$. We only have to simulate the underlying at $\Delta t$. Equation \eqref{eq_stationarity_assumption} holds if the process $(X_{t})_{0\leq t\leq T}$ has stationary increments.\\ 

\textbf{Delta and Gamma as by-product of the method}\\
Generally, the efficiency of the method allows a fast computation of sensitivities via bump and re-run. For Delta and Gamma the polynomial structure of the Chebyshev approximation allows for a direct computation without re-running the time-stepping. Instead we only need to differentiate a polynomial. For Delta we obtain 
\begin{align*}
\frac{\partial V}{\partial x}(x)\approx\sum_{j=0}^{N}c_{j}\frac{\partial p_{j}}{\partial x}(x),
\end{align*} 
which is again a polynomial with degree $N-1$ and for Gamma we obtain
\begin{align*}
\frac{\partial^{2} V}{\partial x^{2}}(x)\approx\sum_{j=0}^{N}c_{j}\frac{\partial^{2} p_{j}}{\partial x^{2}}(x),
\end{align*} 
a polynomial of degree $N-2$.\\

\textbf{Several options on one underlying:}\\
The structure of the dynamic Chebyshev algorithm for exposure calculation exhibits additional benefits for the complex derivative portfolios. For instance, consider non-directional strategies and structured products that offer different levels of capital protection or enhanced exposure. They are typically constructed from a combination of European options, with different strikes and maturities, together with Bermudan options and barrier options. Such structures are essentially a portfolio of derivatives on the same underlying asset, and in this case, the pricing and exposure calculation can be simplified by choosing the same interpolation domain. First, we only need to compute the conditional moments once and  then we can use them for all options. Second, we require less computation in the exposure calculation. Assume we have two options and we are in the time stepping of the Dynamic Chebyshev algorithm at step $t_{u}$. We have two Chebyshev approximations $\widehat{V}^{1}_{t_{u}}=\sum c^{1}_{j}(t_u)p_j$ and $\widehat{V}^{2}_{t_{u}}=\sum c^{2}_{j}(t_u)p_j$. For the exposure calculation we need to compute $\widehat{V}^{1/2}_{t_u}(X_{t_{u}}^{i})=\sum c^{1/2}_{j}(t_u)p_j(X_{t_{u}}^{i})$ for all risk factors $i=1,\ldots,M$. Hence the evaluation of the Chebyshev polynomials $p_{j}$ at the risk factors $X_{t_{u}}^{i}$ is the same and has only to be done once. In summary, with low additional effort, we can calculate the exposure of several options on one underlying. 

\subsection{Implementational aspects of the DC method for exposure calculation}
In this section, we discuss several implementational aspects which can help to achieve a high performance.\\

\textbf{Choice of interpolation domain:}\\
The choice of a suitable interpolation domain is an important step to ensure a high efficiency of the method. In order to do so one can explore additional knowledge of the specific product. In general the choice of the domain is a trade-off between speed (small domain, low number of nodal points) and accuracy (larger domain, more nodal points). We give an idea how to find a domain for three different examples.

First we consider a Bermudan put option. Here we know that the value of the option converges towards zero if the (log-) price of the underlying goes to infinity. The upper bound $\ol{x}$ is therefore no problem and if we have a risk factor with $X^{i}_{t_{u}}>\ol{x}$ we can simply set $V_{t_{u}}(X^{i}_{t_{u}})=0$. For very low values of $x$ the option is always exercised and thus we set $V_{t_{u}}(X^{i}_{t_{u}})=g(X^{i}_{t_{u}})$ if $X^{i}_{t_{u}}<\ul{x}$.

For an European call or put option we can use the Call-Put parity $C_{t}(x)-P_{t}(x)=e^{x}-e^{-r(T-t)}K$ to find a suitable interpolation domain. The price of a call option converges towards zero for small $x$ and towards $e^{x}-e^{-r(T-t)}K$ for large $x$. We choose $\ul{x}$, $\ol{x}$ such that $C_{t}(\ul{x})$ and $P_{t}(\ol{x})$ are sufficiently small. Then we can set $V_{t_{u}}(X^{i}_{t_{u}})=0$ for $X^{i}_{t_{u}}<\ul{x}$ and $V_{t_{u}}(X^{i}_{t_{u}})=e^{X^{i}_{t_{u}}}-e^{-r(T-t)}K$ for $X^{i}_{t_{u}}>\ol{x}$.

As our last example we consider an up-and-out call option with barrier $b$. Here $b$ is the logical upper bound of the interpolation domain and for $\ul{x}$ we proceed similarly to the European call option case.\\

\textbf{Smoothing:}\\
If the payoff of the option has a kink or discontinuity the approximation with Chebyshev polynomials is not efficient. In this case we can modify the algorithm and improve convergence by a "smoothing" of the first time step. We can exploit that the continuation value at $t_{n-1}$ is exactly the value of a European option with duration $\Delta t=t_{n}-t_{n-1}$, i.e.
\begin{align}\label{eq:payoff_smoothing}
V_{t_{n-1}}(x)=\max\{g(x),P^{EU}(x)\} \quad \text{with} \quad P^{EU}(x)=\EE^{\mathbb{Q}}[g(X_{t_{n}})\vert X_{t_{n-1}}=x].
\end{align}
Often, it is more efficient to compute directly the European option price $\EE^{\mathbb{Q}}[g(X_{t_{n}})\vert X_{t_{n-1}}=x_{k}]$ at the nodal points $x_{k}$, $k=0,\ldots,N$. Hence, there is no interpolation error in the first step and we start with the interpolation of the (smooth) function $V_{t_{n-1}}$. We use this technique for all our numerical experiments. The influence of this modification on the error decay is investigated in \cite{GlauMahlstedtPoetz2019}.\\

\textbf{Splitting:}\\
If the value function is not smooth enough and we require a relatively high accuracy the number of Chebyshev nodes can become large. In this case it is often beneficial to split the domain in subdomains and interpolate on each of the subdomains. For a Bermudan option the boundary of the exercise region would be the right splitting point. On each subdomain the option value is now very smooth and much fewer nodal points are needed to achieve the same accuracy. Unfortunately, this approach has the drawback that the conditional expectations can no longer be precomputed. If the underlying process is normally distributed and we haven an analytic solution for the condition expectations this is not a problem. However, if the computation of the conditional expectations is numerically costly the approach becomes too complex and thus inefficient.\\

\section{Numerical experiments}
In this section, we investigate the Dynamic Chebyshev method numerically by calculating the credit exposure for Bermudan and barrier options. We provide evidence that the Dynamic Chebyshev approach is well-suited for the exposure calculation and investigate how the runtime of the method behaves when we increase the number of simulations. More precisely, we compute the expected exposure and the potential future exposure for a Bermudan put option and an up-and-out barrier option in three different asset models. The resulting exposure profiles over the lifetime of the options are displayed and we discuss their plausibility. As models for the underlying risk factor we use the Black-Scholes model (diffusion), the Merton model (jump-diffusion) and the CEV model (local volatility). In each of the models we exploit a different technique to calculate the generalized moments in order to show the flexibility of the Dynamic Chebyshev approach. Moreover, we run the method for different numbers of simulation paths and investigate the runtimes of the three phases of the method, i.e. simulation, pre-computation and time-stepping. 

\subsection{Problem description}
We consider two option pricing problems. First, we choose a Bermudan put option with maturity $T$ and exercise dates $t_{u}$, $u=0,\ldots,n$ equally distributed between $0$ and $T$. This yields for the value function
\begin{align*}
V_{T}(x)&=(K-e^{x})^{+}\\
V_{t_{u}}(x)&=\max\left\{(K-e^{x})^{+},e^{-r(t_{u+1}-t_{u})}\EE[V_{t_{u+1}}(X_{t_{u+1}})\vert X_{t_{u}}=x]\right\}
\end{align*}
with strike $K$. 

Second, we consider a discretely monitored barrier up-and-out call option with maturity $T$. We assume that the barrier is monitored on dates $t_{u}$, $u=0,\ldots,n$ equally distributed between $0$ and $T$. The pricing problem becomes
\begin{align*}
V_{T}(x)&=(e^{x}-K)^{+}\1_{(-\infty,b]}(x)\\
V_{t_{u}}(x)&=\EE[V_{t_{u+1}}(X_{t_{u+1}})\vert X_{t_{u}}=x]\1_{(-\infty,b]}(x)
\end{align*}
with strike $K$, barrier $B$ and $b=\log(B)$. 

We fix for both options a strike of $K=100$, initial stock price $S_{0}=100$ and interest rate $r=0.03$. Moreover, we fix the number of exercise dates (Bermudan) or monitor days (barrier) as $n=52$ which is a weekly exercise schedule and use the same days for the exposure calculation. For the experiments we introduce the following three asset price models and explain how we compute the corresponding generalized moments.\\

\textbf{The Black-Scholes model}\\
In the classical model of \cite{BlackScholes1973} the stock price process is modelled by the SDE
\begin{align*}
\text{d}S_{t}=\mu S_{t}\text{d}t + \sigma S_{t}\text{d}W_{t}.
\end{align*}
with drift $\mu$ and volatility $\sigma>0$ under the real-world measure $\mathbb{P}$. Under the pricing measure $\mathbb{Q}$ the drift equals $r$. Exploiting the fact that the log-returns $X_{t}=\log(S_{t}/S_{0})$ are normally distributed we obtain an analytic formula for the generalized moments $\Gamma_{k,j}$. As model parameter we fix  volatility $\sigma=0.25$ and drift $\mu=0.1$.\\

\textbf{The Merton jump diffusion model}\\
The jump diffusion model introduced by \cite{Merton1976} adds jumps to the classical Black-Scholes model. The log-returns follow a jump diffusion with volatility $\sigma$ and added jumps arriving at rate $\lambda>0$ with normal distributed jump sizes according to $\mathcal{N}(\alpha,\beta^{2})$. The stock price under $\mathbb{P}$ is modelled by the SDE
\begin{align*}
\text{d}S_{t}=\mu S_{t}\text{d}t + \sigma S_{t}\text{d}W_{t} + \text{d}J_{t}
\end{align*}
for a compound Poisson process $J_{t}$ with rate $\lambda$. The characteristic function of the log-returns $X_{t}=\log(S_{t}/S_0)$ under the pricing measure $\mathbb{Q}$ is given by
\begin{align*}
\varphi(z)=exp\left(t\left(ibz - \frac{\sigma^{2}}{2}z^{2} + \lambda\left(e^{iz\alpha - \frac{\beta^{2}}{2}z^{2}}-1\right)\right)\right)
\end{align*} 
with risk-neutral drift
\begin{align*}
b=r-\frac{\sigma^{2}}{2}-\lambda\left(e^{\alpha+\frac{\beta^{2}}{2}}-1\right).
\end{align*}
In our experiments we calculate the conditional expectations $\Gamma_{k,j}$ using numerical integration and the Fourier transforms of the Chebyshev polynomials along with the characteristic function of $X_t$. We fix the parameters 
\[\sigma=0.25,\quad \alpha=-0.5,\quad \beta=0.4,\quad \lambda=0.4\quad \text{and} \quad  \mu=0.1.\]

\textbf{The Constant Elasticity of Variance model}\\
The Constant Elasticity of Variance model (CEV) as stated in \cite{Schroder1989} is a local volatility model based on the stochastic process
\begin{align}\label{CEV_model_SDE}
\text{d}S_{t}=\mu S_{t}\text{d}t + \sigma S_{t}^{\beta/2}\text{d}W_{t} \qquad \text{for} \qquad \beta>0.
\end{align}
Hence the stock volatility $\sigma S_{t}^{(\beta-2)/2}$ depends on the current level of the stock price. For the special case $\beta=2$ the model coincides with the Black-Scholes model. However, from market data one typically observes a $\beta<2$. The CEV-model is one example of a model which has neither a probability density, nor a characteristic function in closed-form. Hence, Monte-Carlo simulation is utilised to calculate the conditional expectations $\Gamma_{k,j}$. This means $X_{\Delta t}$ has to be simulated under $\mathbb{Q}$ for different starting values $X_{0}=x_{k}$, $k=0,\ldots,N$. For our experiments we fix the following parameters
\[\sigma=0.3,\quad \beta=1.5\quad \text{and} \quad  \mu=0.1.\]

Based on the chosen numerical techniques we expect that the pre-computation phase in the Black-Scholes model is the fastest due to the analytic formula and the one in the CEV model is the slowest. 

\subsection{Credit exposure of Barrier options}
In this section, we investigate the expected exposure and the potential future exposure of a barrier option. We consider a call option with up-and-out barrier $B=150$ in the Black-Scholes and Merton jump-diffusion model. In the CEV model we choose a lower barrier of $B=125$ because of the lower volatility for higher stock prices and the absence of jumps. We fix the interpolation domain $\mathcal{X}=[\ul{x},\ol{x}]$ with $\ul{x}=\log(10)$ and $\ol{x}=\log(B)$. Moreover, we fix $N=40$ Chebyshev points. We price the option and calculate the exposure for an increasing number of simulations paths of the underlying risk factor.

\subsubsection{Exposure profiles}
Figure \ref{fig:EE_Barrier} shows the expected exposure (EE) in the three models over the option's lifetime and for a fix number of simulation paths. We observe a different behaviour in each of the three models. The expected exposure increases slightly in the Black-Scholes model, decreases slightly in the Merton model and increases in the CEV model. Two different effects have an influence on the expected exposure. On the one hand, the option is an at-the-money call option and due to the positive drift of the underlying (under the real-world measure) we expect that over the option's lifetime more and more paths will be in the money which leads to an increase of the expected exposure. On the other hand, when the underlying increases there is a higher risk that the underlying reaches the barrier and the option is knocked out. Thus the expected exposure should decay. Depending on the model properties one of the effects might dominate the other. In the presence of jumps, as in the Merton model, there is an additional risk that the price of the underlying jumps over the barrier. The risk increases over time with a positive drift. This explains the decreasing exposure over time. In the CEV model however, the volatility decreases with a higher stock price and thus very large upward movements are less likely. Therefore, the option's exposure increases due to the positive drift of the underlying risk factor. The Black-Scholes model has also no jumps but a constant volatility which explains that it lies between the two other models.

Figure \ref{fig:PFE_Barrier} shows the corresponding potential future exposure (PFE) at the $97,5\%$ level in the three models. Here, we observe a more similar behaviour in the models. The potential future exposure increases over time due to the positive drift and the diffusion term in all three models. The PFE converges towards the maximal possible exercise value $B-K$ which is $50$ in the Black-Scholes and Merton model and $25$ in the CEV model.

Already this simplified toy example makes it evident that the expected exposure profile of a barrier option critically depends on the model choice. Therefore, it is essential to have a method for the exposure calculation which is \textit{flexible} in the model choice. In particular, one needs to be able to handle different model features such as jumps and local volatility. As in practice a variety of models will be used to quantify the credit exposure a method which allows to easily switch between different models is desirable. As the experiments show the Dynamic Chebyshev method for exposure calculation exhibits this flexibility in the models.

\begin{figure}[H] 
\centering
\includegraphics[width=\textwidth]{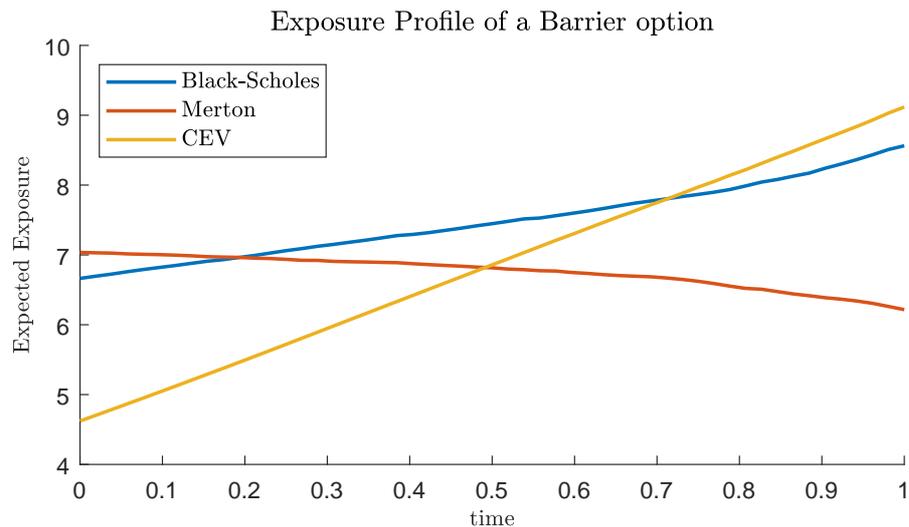} 
\caption{Expected exposure (EE) for a barrier option in the Black-Scholes model, the Merton model and the CEV model with maturity $T=1$ and $N_{sim}=50000$.}
  \label{fig:EE_Barrier} 
\end{figure}

\begin{figure}[H] 
\centering
\includegraphics[width=\textwidth]{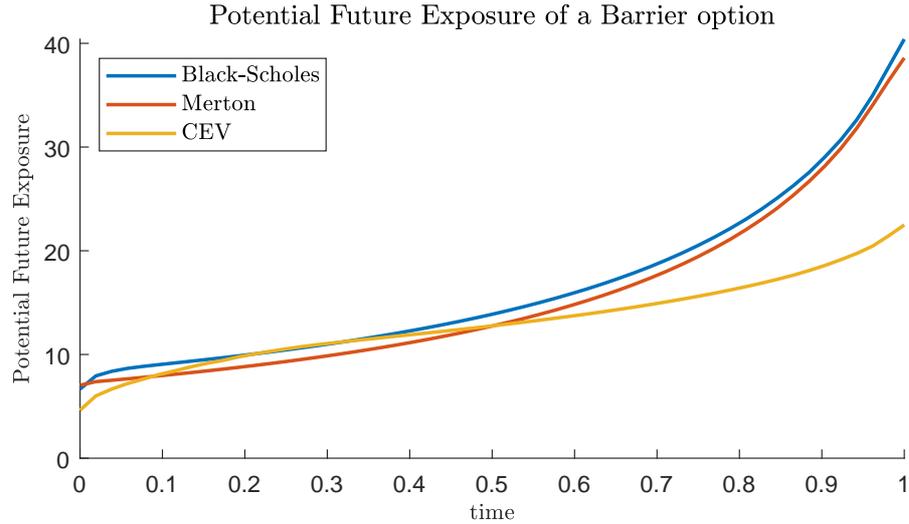} 
\caption{Potential future exposure (PFE) for a barrier option in the Black-Scholes model, the Merton model and the CEV model with maturity $T=1$ and $N_{sim}=50000$.}
  \label{fig:PFE_Barrier} 
\end{figure}

\subsubsection{Runtimes}
Figure \ref{fig:EE_Barrier_runtime} shows the runtime as function of the number of simulations in the three models and Table \ref{tab:EE_Barrier_rumtime} displays the corresponding runtimes. The total runtime of the exposure calculation is split into the runtimes of the simulation phase, the pre-computation step and the time-stepping. We observe that the runtime of the pre-computation step varies between the three models but is independent of the number of simulations. The analytic formula in the Black-Scholes model is the fastest with less than $0.01s$, then we have the Fourier approach (Merton model) with $0.02s$ and the Monte-Carlo approach (CEV model) with $0.34s$. For all three models the runtime of the pre-computation step is feasible. As an additional benefit, model specific information such as the characteristic function or the density can be exploited to improve efficiency. Moreover, we observe that the time-stepping is model independent, the runtimes are the same in all three cases and increase in the number of simulations. For the simulation phase we observe the model dependence or more precisely the dependence on the technique used for the simulation of the risk factors, which reflects our model specific knowledge. 

\begin{figure}[H]
\begin{minipage}{.5\linewidth}
\centering
\subfloat[]{\includegraphics[scale=.48]{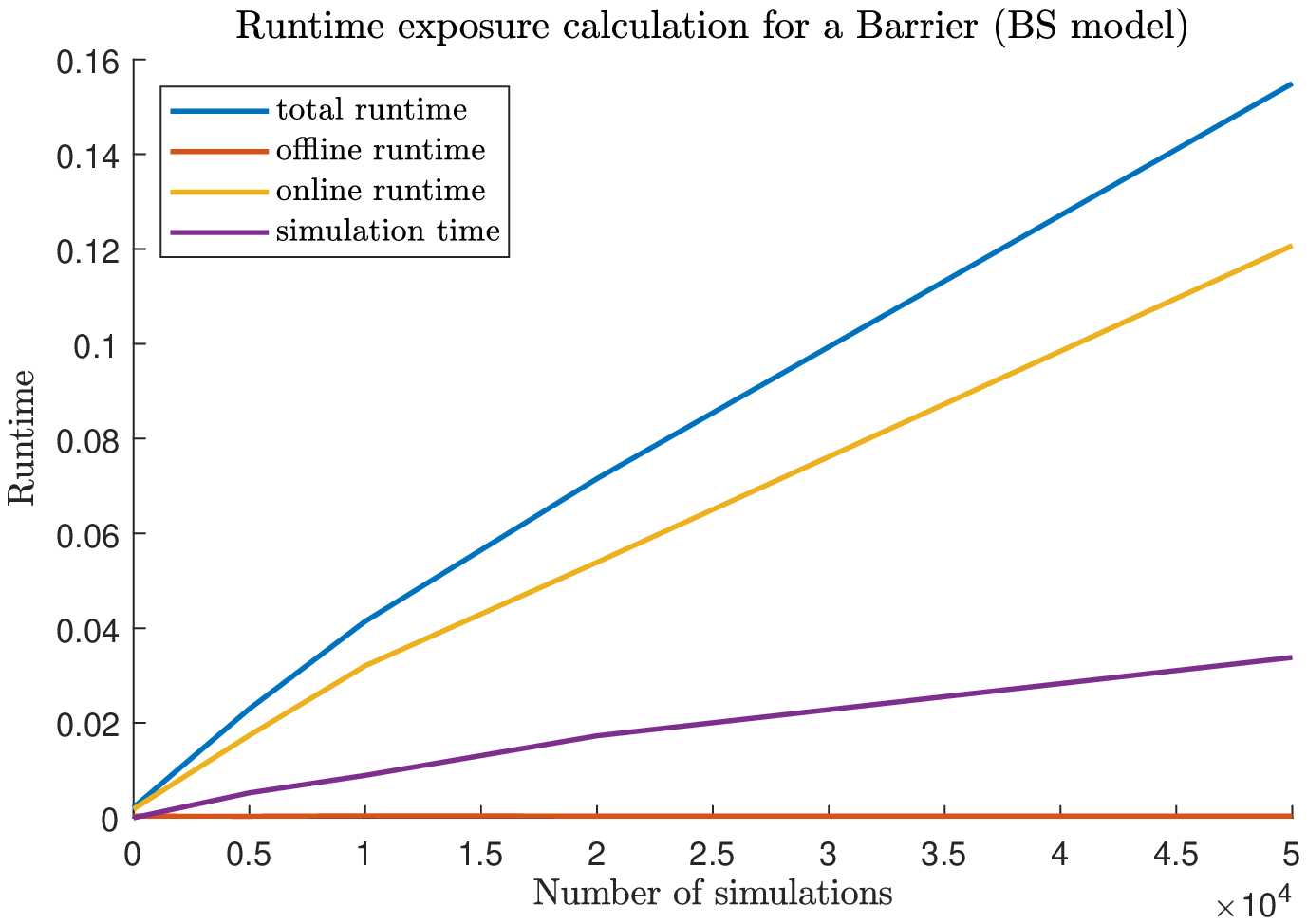}}
\end{minipage}%
\begin{minipage}{.5\linewidth}
\centering
\subfloat[]{\includegraphics[scale=.48]{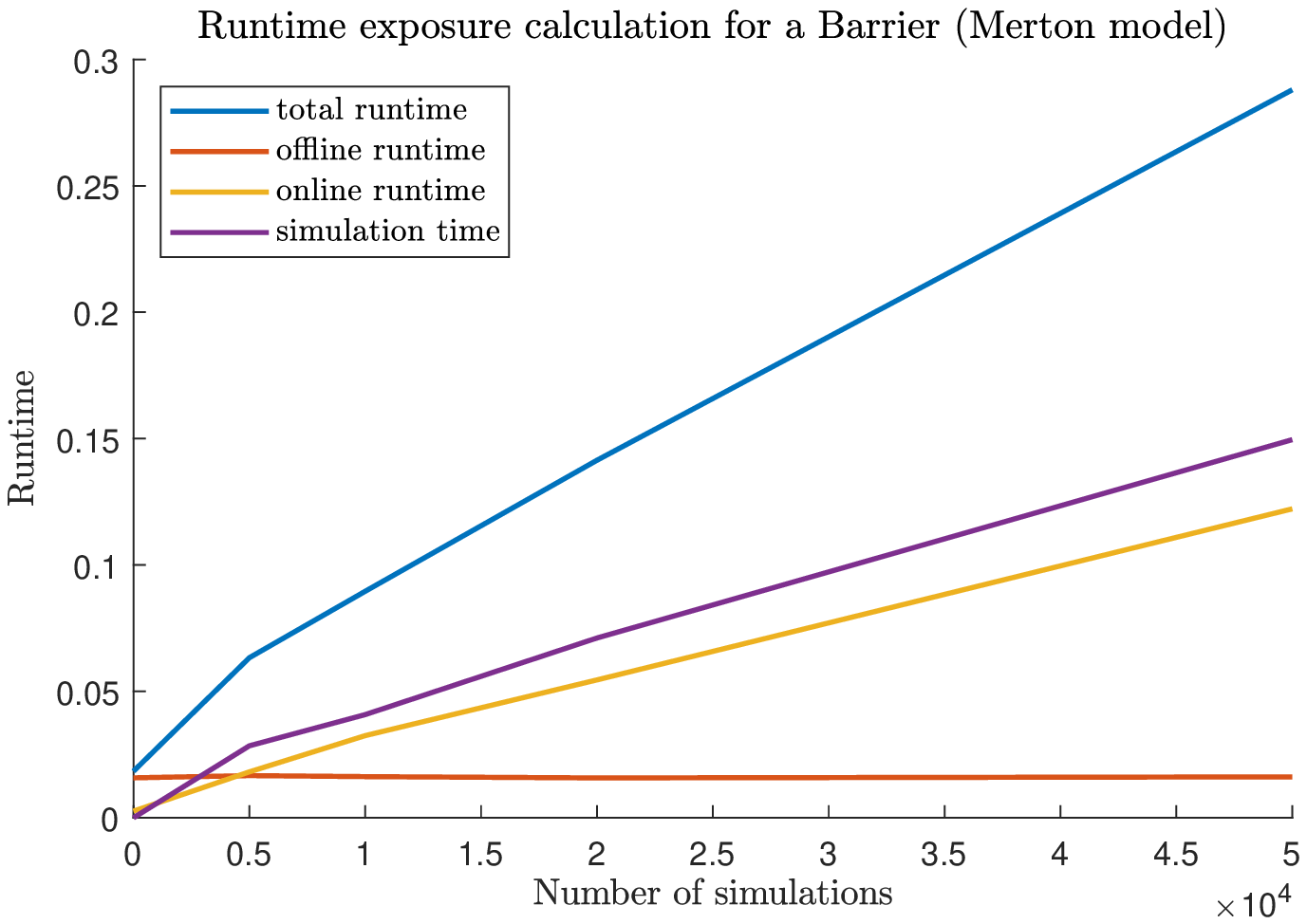}}
\end{minipage}\par\medskip
\centering
\subfloat[]{\includegraphics[scale=.48]{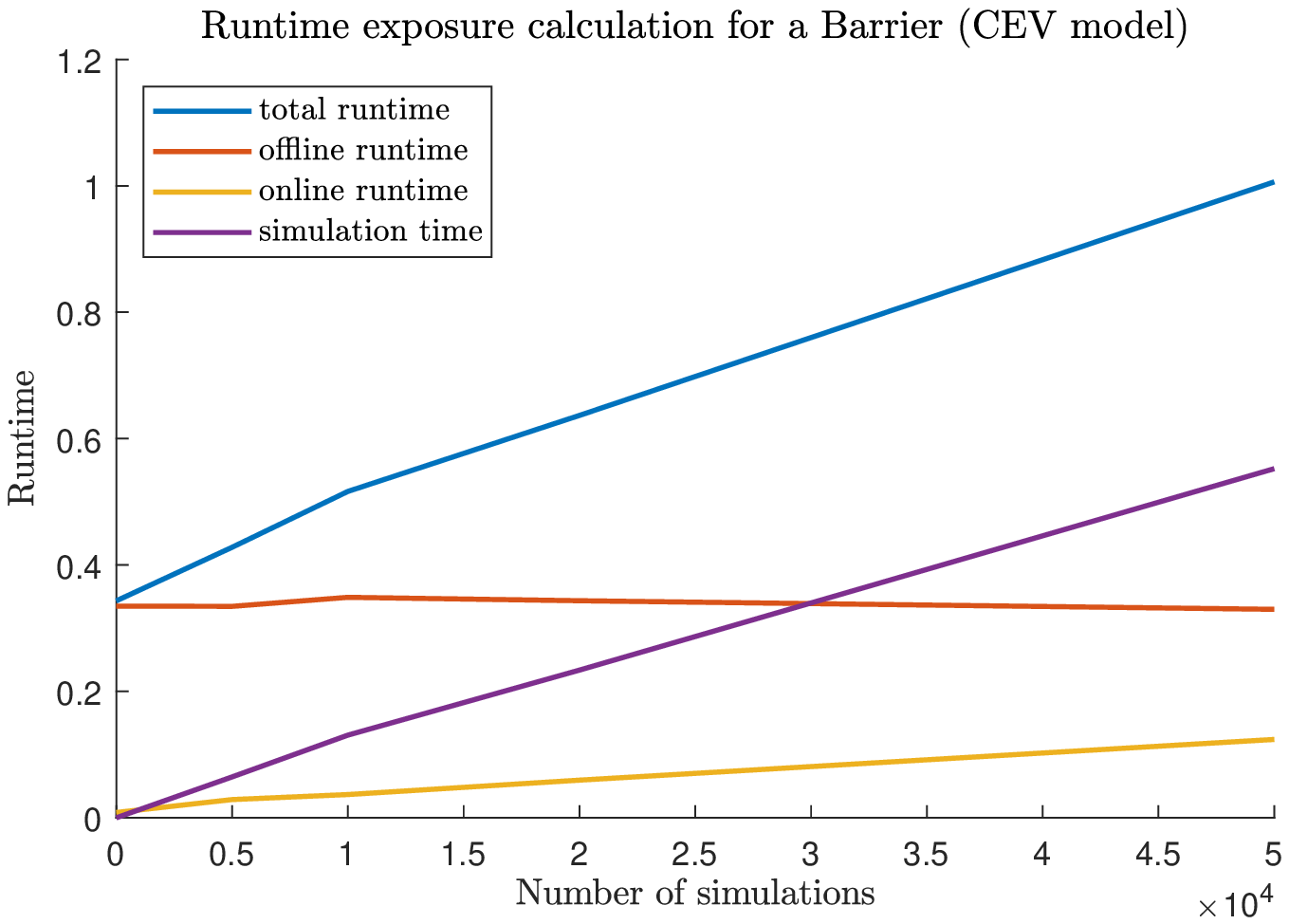}}
\caption{Runtime of the credit exposure calculation for a barrier option in the Black-Scholes model, the Merton model and the CEV model for different numbers of simulation paths. All three plots display the total runtime (blue), the runtime of the pre-computation step (red), the runtime of the time-stepping (yellow) and the runtime of the simulation step (purple).}
\label{fig:EE_Barrier_runtime}
\end{figure}

\begin{table}[H]
\begin{center}
\begin{tabular}{llccccc}
\hline 
 &$\qquad N_{sim}$ & $0$ & $5000$ & $10000$ & $20000$ & $50000$ \\ 
\hline 
\multirow{4}{*}{BS}
& Simulation & 0s & 0.01s & 0.01s & 0.02s & 0.04s\\ 
& Pre-computation & $<$0.01s & $<$0.01s & $<$0.01s & $<$0.01s & $<$0.01s\\
& Time-stepping & $<$0.01s & 0.03s & 0.04s & 0.07s & 0.14s\\ 
& Total & $<$0.01s & 0.04s & 0.05s & 0.09s & 0.18s\\  
\hline 
\multirow{4}{*}{Merton}
& Simulation & 0s & 0.02s & 0.05s & 0.06s & 0.15s\\ 
& Pre-computation & 0.02s & 0.02s & 0.02s & 0.02s & 0.02s\\
& Time-stepping & $<$0.01s & 0.03s & 0.04s & 0.06s & 0.12s\\ 
& Total & 0.02s & 0.07s & 0.11s & 0.15s & 0.30s\\  
\hline
\multirow{4}{*}{CEV}
& Simulation & 0s & 0.07s & 0.13s & 0.26s & 0.58s\\ 
& Pre-computation & 0.34s & 0.35s & 0.35s & 0.35s & 0.34s\\
& Time-stepping & 0.01s & 0.02s & 0.04s & 0.07s & 0.14s\\ 
& Total & 0.35s & 0.44s & 0.52s & 0.68s & 1.06s\\  
\hline 
\end{tabular} 
\caption{Runtime of the credit exposure calculation for a barrier option in the Black-Scholes (BS) model, the Merton model and the CEV model for different numbers of simulation paths.}
\label{tab:EE_Barrier_rumtime} 
\end{center}
\end{table}

\subsection{Credit exposure of Bermudan options}
We price a Bermudan option with the Dynamic Chebyshev algorithm and calculate the expected exposure \eqref{eq:EE} and the potential future exposure \eqref{eq:PFE}. We consider a Bermudan put option in the Black-Scholes, the Merton jump-diffusion model and in the CEV model. We fix the interpolation domain $\mathcal{X}=[\ul{x},\ol{x}]$ with $\ul{x}=\log(0.2)$ and $\ol{x}=\log(350)$. Moreover, we fix $N=150$ Chebyshev points. We price the option and calculate the exposure for an increasing number of simulation paths of the underlying risk factor.

\subsubsection{Exposure profiles}
We expect to observe a decreasing expected exposure over time due to the early-exercise feature of the Bermudan option. We know that if the option is exercised the exposure becomes zero. As we approach maturity, the optimal exercise boundary of the Bermudan put converges towards the option's strike. Therefore, the likeliness that the option is exercised early if the underlying is below the strike increases over the option's lifetime. Similarly, the likeliness that the option ends in-the-money if the underling is above the strike decreases over time. Both effects yield a decreasing exposure over the lifetime of the option. For the potential future exposure we expect to see two different effects. First, the diffusion term should lead to an increasing PFE whereas the early-exercise feature still leads to a decreasing PFE over time.

Figure \ref{fig:EE_Bermudan} and \ref{fig:PFE_Bermudan} show the expected exposure and the potential future exposure for a level of $\alpha=0.975$. We observe that the expected exposure decreases over the lifetime of the option in all three different models. Close to maturity the option is for most paths either already exercised or it is out-of-the-money. The potential future exposure first increases and then decreases in all three models. In both cases do the profiles of the Black-Scholes and Merton model behave similarly. The profiles of the CEV model exhibit however a slightly different shape. This indicates that the additional jumps in the Merton model have less striking influence on the behaviour of the exposure profiles than for a barrier option. The study shows us that the Dynamic Chebyshev method enables us to calculate the exposure of Bermudan options in different models.
\begin{figure}[H] 
\centering
\includegraphics[width=\textwidth]{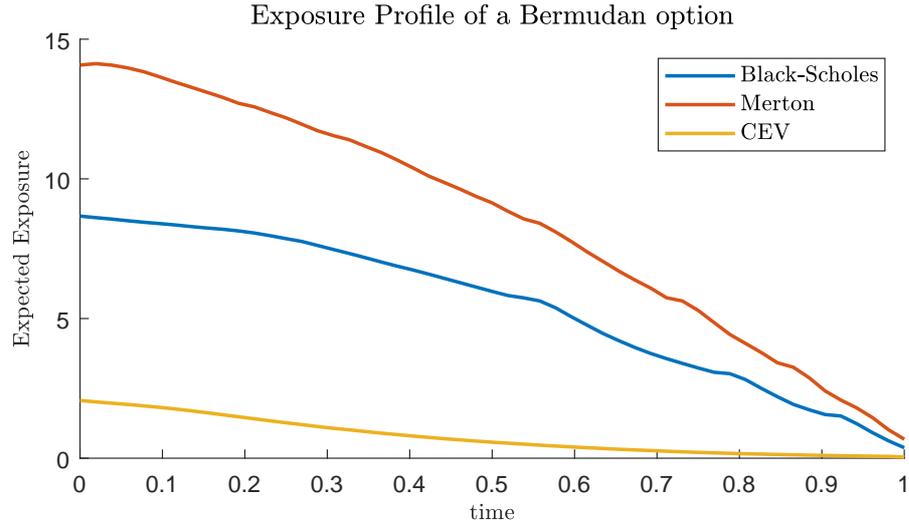} 
\caption{Expected exposure (EE) for a Bermudan option in the Black-Scholes model, the Merton model and the CEV model with maturity $T=1$ and $N_{sim}=50000$.}
  \label{fig:EE_Bermudan} 
\end{figure}

\begin{figure}[H] 
\centering
\includegraphics[width=\textwidth]{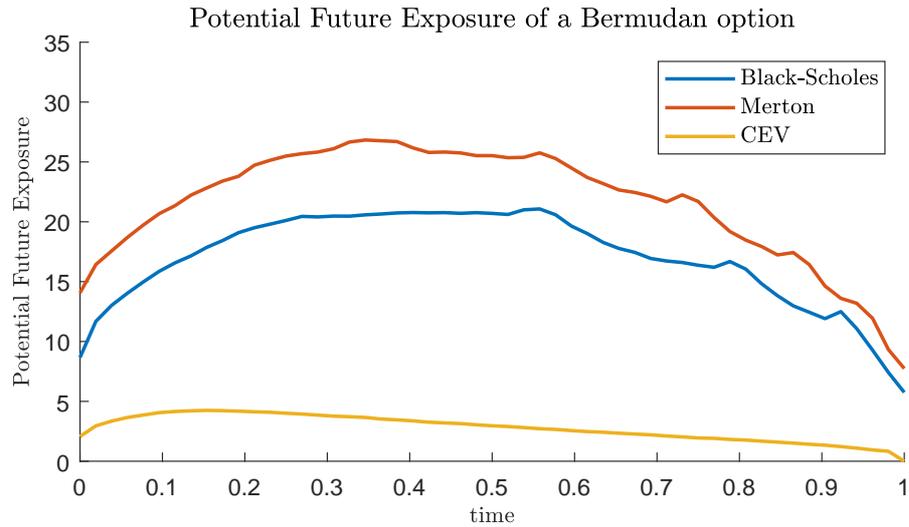} 
\caption{Potential future exposure (PFE) for a barrier option in the Black-Scholes model, the Merton model and the CEV model with maturity $T=1$ and $N_{sim}=50000$.}
  \label{fig:PFE_Bermudan} 
\end{figure}

\subsubsection{Runtimes}
The runtimes are shown in Figure \ref{fig:EE_Bermudan_runtime} as a function of the number of simulations and they are displayed in Table \ref{tab:EE_Bermudan_rumtime}. The runtime of the pre-computation step and of the time-stepping are slower as for the barrier option due to the higher number of Chebyshev points $N$. As a function of the number of simulations the runtimes for the Bermudan option and the runtimes for the barrier option behave similarly. The simulation of the risk factors is independent of the number of nodal points and thus the same in both cases. The results for Bermudan options are consistent with the promising results for barrier options from the previous section. This give rise to the hope that we obtain similar results also in further applications.

\begin{figure}[H]
\begin{minipage}{.5\linewidth}
\centering
\subfloat[]{\includegraphics[scale=.48]{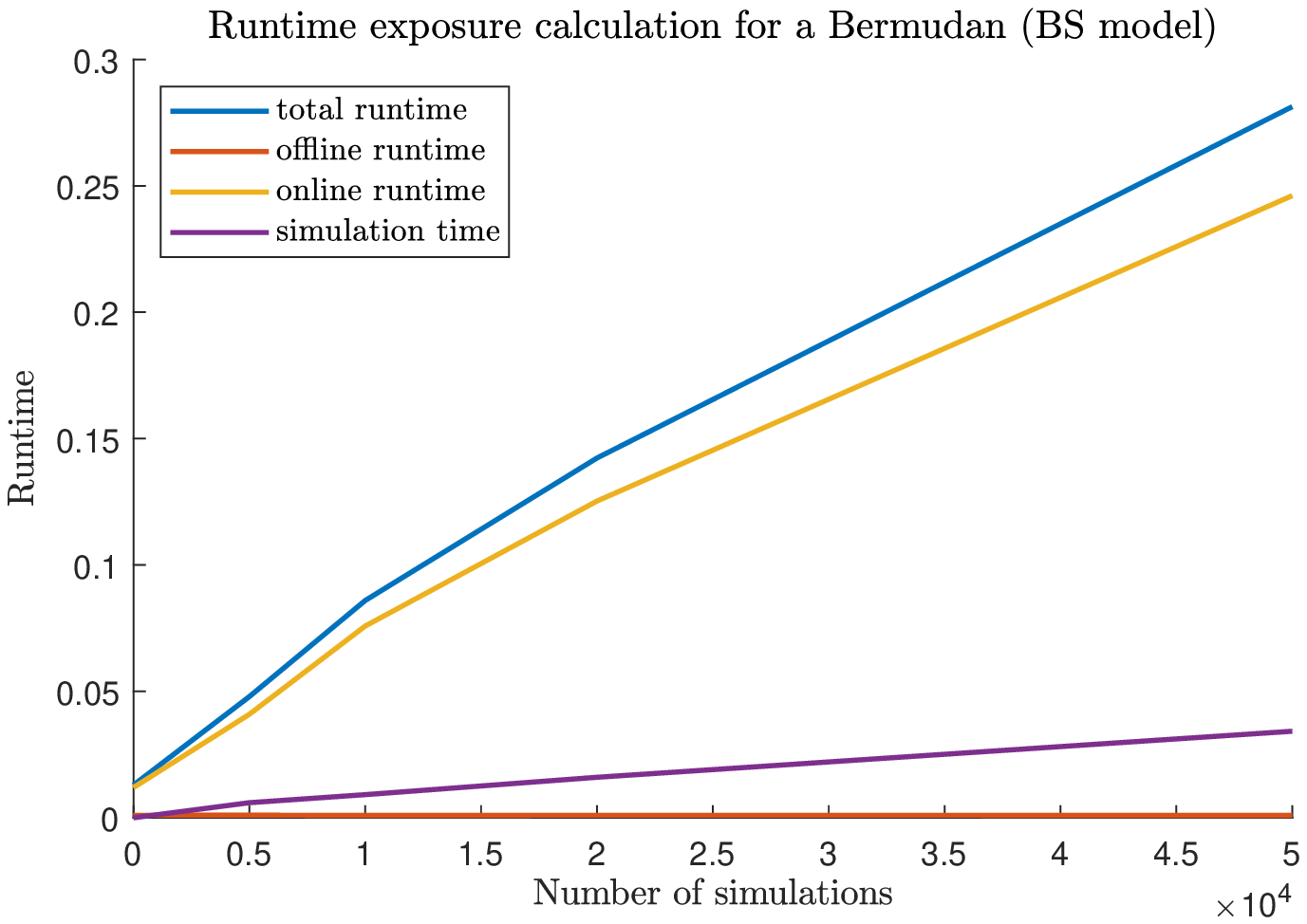}}
\end{minipage}%
\begin{minipage}{.5\linewidth}
\centering
\subfloat[]{\includegraphics[scale=.48]{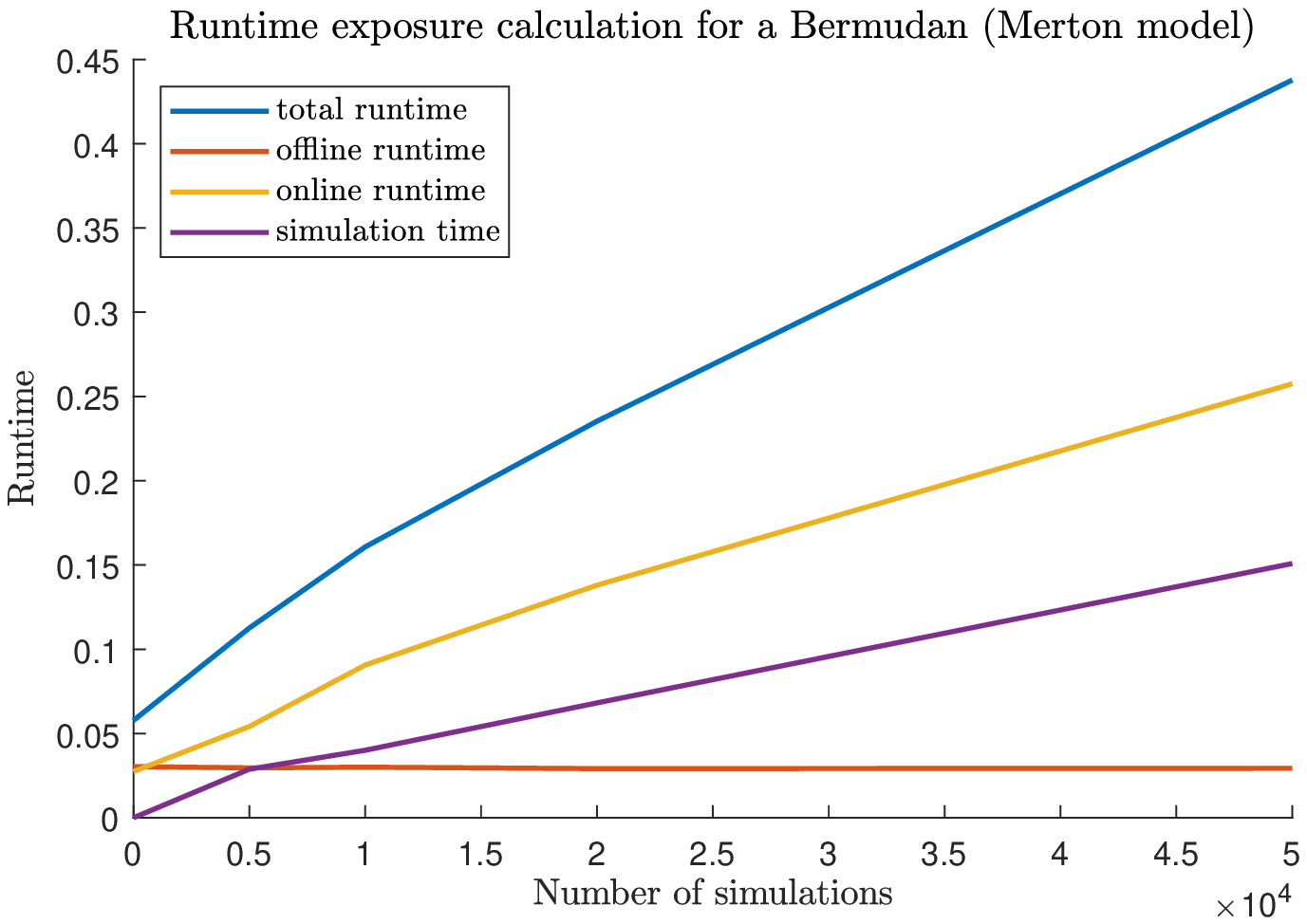}}
\end{minipage}\par\medskip
\centering
\subfloat[]{\includegraphics[scale=.48]{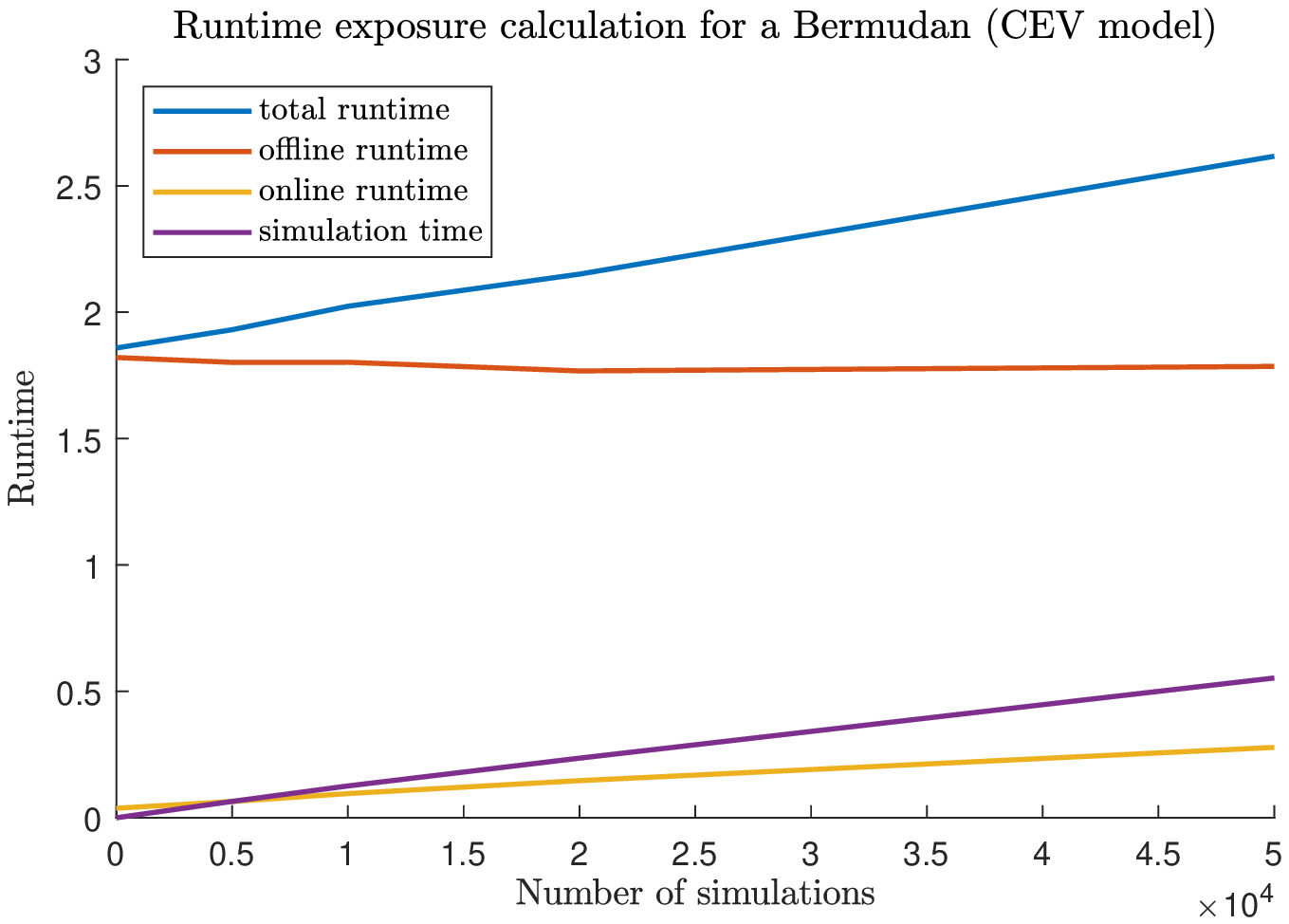}}
\caption{Runtime of the credit exposure calculation for a Bermudan option in the Black-Scholes model, the Merton model and the CEV model for different number of simulation paths. All three plots display the total runtime (blue), the runtime of the pre-computation step (red), the runtime of the time-stepping (yellow) and the runtime of the simulation step (purple).}
\label{fig:EE_Bermudan_runtime}
\end{figure}

\begin{table}[H]
\begin{center}
\begin{tabular}{llccccc}
\hline 
 &$\qquad N_{sim}$ & $0$ & $5000$ & $10000$ & $20000$ & $50000$ \\ 
\hline 
\multirow{4}{*}{BS}
& Simulation & 0s & 0.01s & 0.01s & 0.02s & 0.04s\\ 
& Pre-computation & $<$0.01s & $<$0.01s & $<$0.01s & $<$0.01s & $<$0.01s\\
& Time-stepping & 0.01s & 0.05s & 0.10s & 0.16s & 0.30s\\ 
& Total & 0.01s & 0.06s & 0.12s & 0.18s & 0.34s\\   
\hline 
\multirow{4}{*}{Merton}
& Simulation & 0s & 0.02s & 0.04s & 0.07s & 0.15s\\ 
& Pre-computation & 0.03s & 0.03s & 0.04s & 0.04s & 0.04s\\
& Time-stepping & 0.03s & 0.06s & 0.12s & 0.17s & 0.31s\\ 
& Total & 0.06s & 0.12s & 0.19s & 0.27s & 0.49s\\  
\hline
\multirow{4}{*}{CEV}
& Simulation & 0s & 0.07s & 0.14s & 0.24s & 0.58s\\ 
& Pre-computation & 1.78s & 1.79s & 1.82s & 1.76s & 1.73s\\
& Time-stepping & 0.04s & 0.07s & 0.13s & 0.17s & 0.31s\\ 
& Total & 1.82s & 1.93s & 2.08s & 2.17s & 2.62s\\   
\hline 
\end{tabular} 
\caption{Runtime of the credit exposure calculation for a Bermudan option in the Black-Scholes (BS) model, the Merton model and the CEV model for different number of simulation paths.}
\label{tab:EE_Bermudan_rumtime} 
\end{center}
\end{table}

\subsection{Summary of the experiments}
In this section, we analysed the Dynamic Chebyshev method for credit exposure calculation numerically.

\cite{GlauMahlstedtPoetz2019} have validated the method for the pricing of options in different asset models. The experiments of this section show that the method is moreover well suited for credit exposure calculation of Bermudan and barrier options. Our examples show that the method can be applied to different models which require different numerical techniques for the moment calculation. This model choice does not affect the time-stepping and thus the exposure calculation itself as confirmed in the experiments. 

Furthermore, the decomposition of the method into a pre-computation step and a time-stepping makes the method highly efficient. The time-stepping includes the pricing and the exposure calculation and depends only on the number of nodal points and the number of simulations. It is independent of the model choice and the numerical technique used in the pre-computation step. The investigation indicates additional efficiency benefits from the method when it is used to compute the exposure for different options on the same underlying.

The analysis of the exposure profiles reveals critical dependence on the model choice. The effect is most striking for the expected exposure of barrier options, where we observed increasing and decreasing profiles for the same option in different models.

\section{Empirical investigation of credit exposure profiles}
In this section, we investigate the exposure profiles of Bermudan and barrier options further and compare them with the one's of European options. Since the calculation of the exposure for complex derivatives with the traditional Monte Carlo method can be prohibitively time-consuming, practitioners sometimes approximate their dynamics with European options, whose profiles are simpler to obtain and well understood. We compare the exposure profiles of Bermudan and barrier options with the profiles of European options and the check the appropriateness of the approach. In the case of a Bermudan option we additionally investigate the effect of the number of exercise rights on the exposure profile.

\subsection{Barrier and European options}
In this section, we compare the exposure profile of barrier options with the exposure profiles of European options. Figure \ref{fig:EE_Barrier_EU} shows the comparison in the Black-Scholes model, in the Merton model and in the CEV model. We choose the same parameters as in the previous section. The expected exposure and the potential future exposure in all three models is smaller for the barrier option due to the knock-out feature. For both option types the expected exposure behaves relatively similar and moves more or less parallel. For the potential future exposure we observe different behaviours over time. For the European case the diffusion term results in an increasing exposure over time. This effect is less strong in the barrier case due to the risk of a knock-out. Closer to maturity the risk of reaching the barrier becomes smaller and the PFE increases faster than its European counterpart. We conclude that replacing barrier by European options yields an overestimation of the credit exposure. This is an indication for a very conservative practice. In the potential future exposure case the difference is substantial. Here, the dynamic Chebyshev method yields more precise result which could lead to lower capital requirements.

\begin{figure}[H]
\begin{minipage}{.5\linewidth}
\centering
\subfloat[]{\includegraphics[scale=.48]{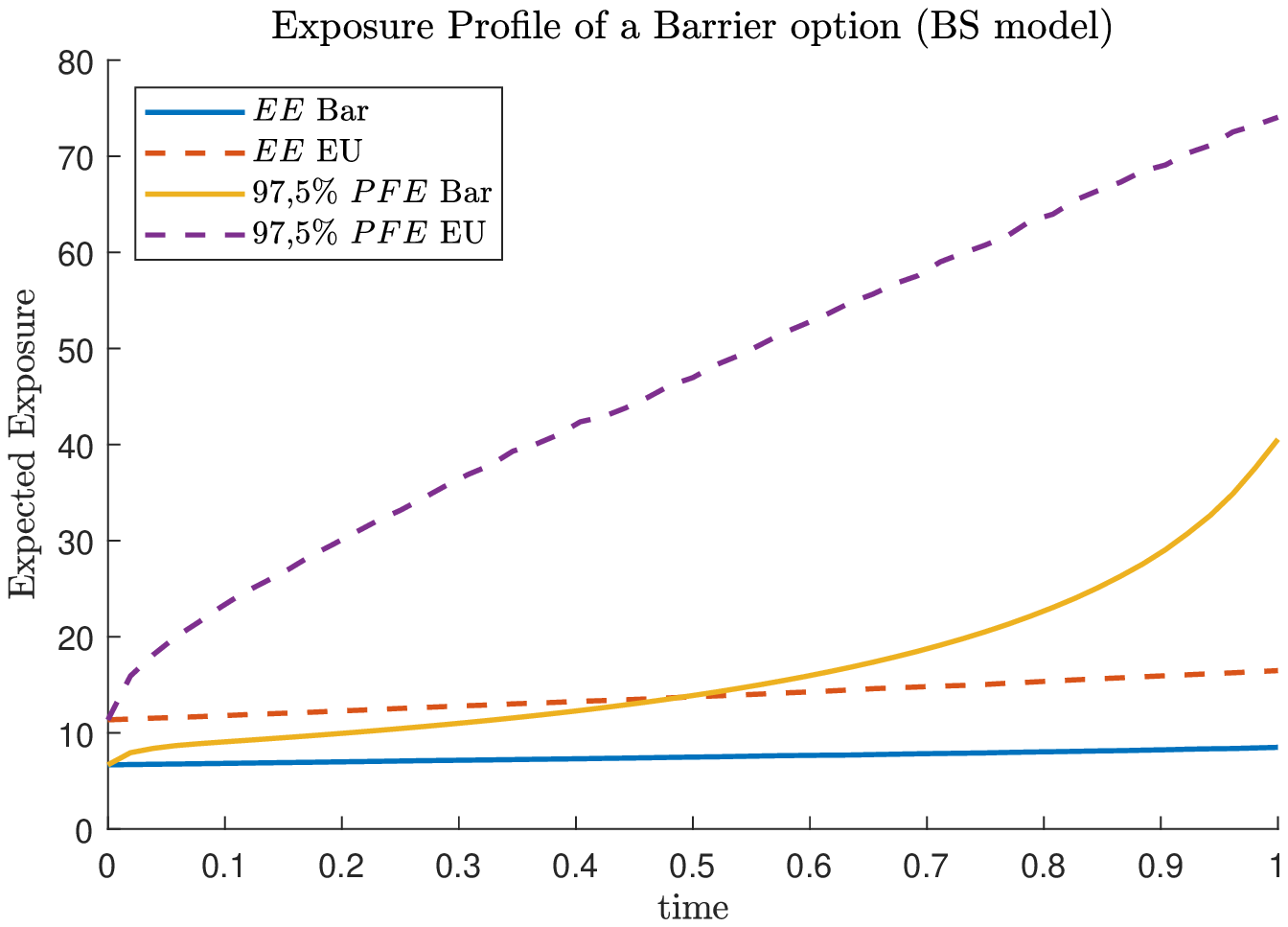}}
\end{minipage}%
\begin{minipage}{.5\linewidth}
\centering
\subfloat[]{\includegraphics[scale=.48]{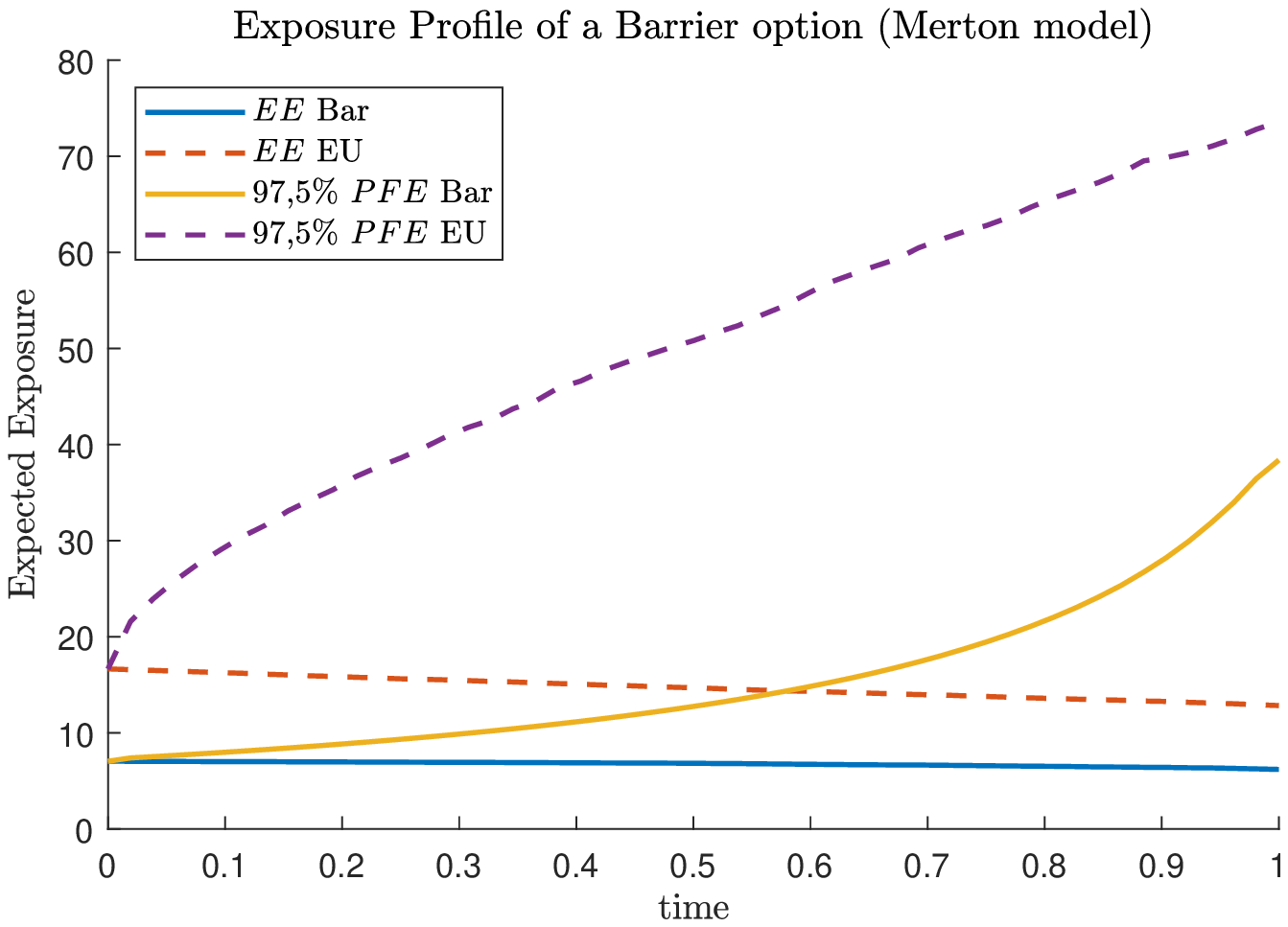}}
\end{minipage}\par\medskip
\centering
\subfloat[]{\includegraphics[scale=.48]{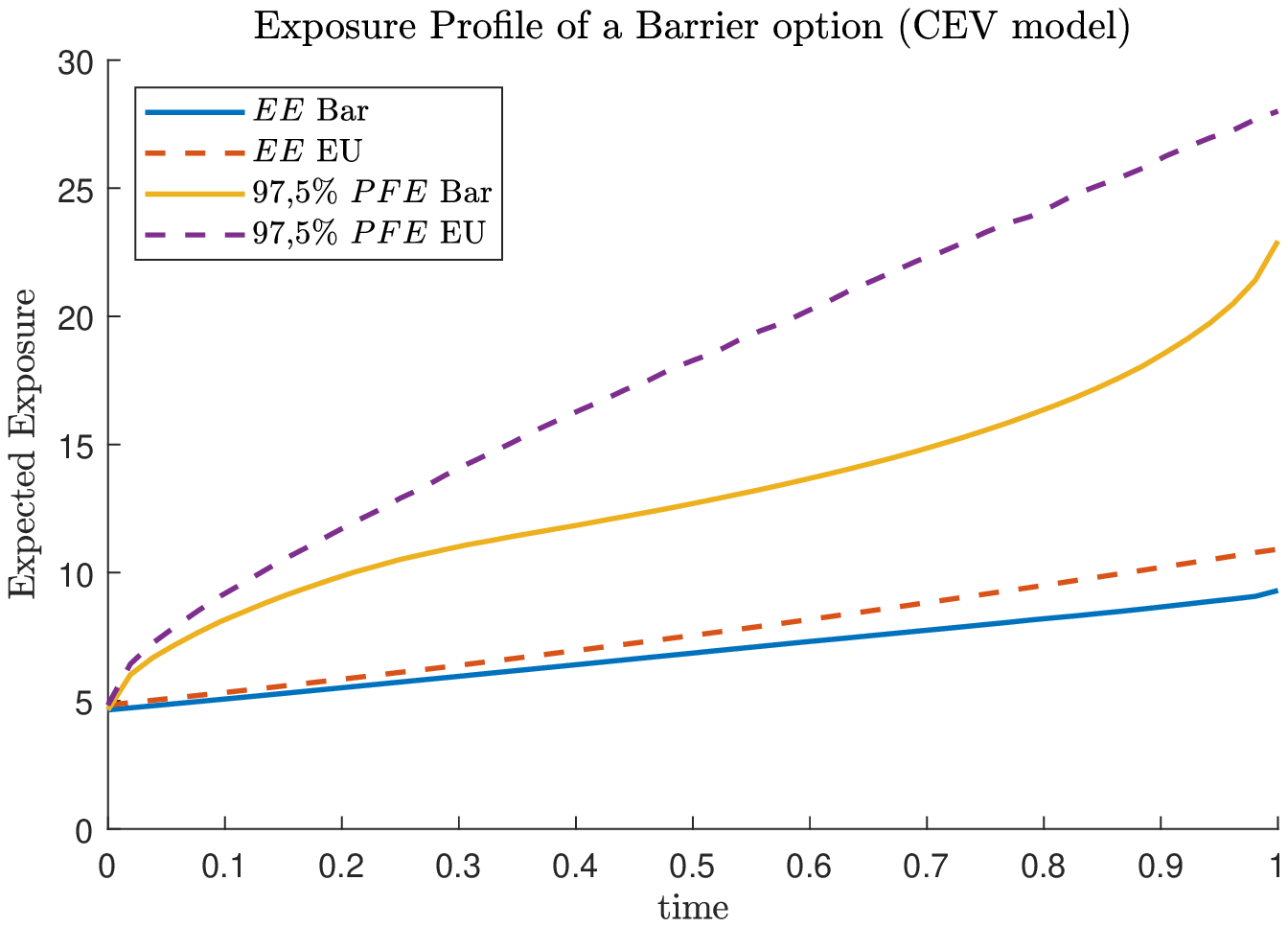}}
\caption{Expected exposure (EE) and potential future exposure (PFE) profiles for a barrier option and a European call option with maturity $T=1$ in the Black-Scholes model, the Merton model and the CEV model with $N_{sim}=50000$.}
\label{fig:EE_Barrier_EU}
\end{figure}

\subsection{Bermudan and European options}
In this section, we compare the exposure profile of Bermudan options with the one of European options. Figure \ref{fig:EE_Bermudan_EU} shows the exposure profiles in the Black-Scholes, Merton and CEV model. The corresponding option values are displayed in Table \ref{tab:EE_Bermudan_EU}. In the experiments we used the same parameter specifications as specified in the previous section.\\

We observe that the price at $t=0$ and thus the expected exposure at $t=0$ of the Bermudan option is higher than the value of the European option in all three models. This is in line with the theory since the possibility to exercise early is an additional feature of Bermudan options compared to European options. With this in mind the difference between the prices can be seen as the added value of the early exercise possibility. Over the option's lifetime however, we observe that the expected exposure of a Bermudan option decreases faster than its European counterpart. The reason for this behaviour is that if a Bermudan option is exercised early the exposure vanishes subsequently. Over its lifetime, the option is exercised for an increasing number of simulation paths and thus the exposure decreases. In contrast, a European option lacks this early exercise feature and has therefore a higher exposure at maturity. In our experiments this difference is about one order of magnitude. For the potential future exposure we observe the same effects. In absolute terms the effect is stronger as for the expected exposure, in relative terms it is slightly lower.

We conclude that replacing a Bermudan by a European option in the exposure calculation has two different effects. First, for most of the option's lifetime it highly overestimates both the option's expected exposure and potential future exposure and seems therefore to be too conservative. Secondly, close to $t=0$ the exposure of the European option underestimates the Bermudan option's exposure slightly. On a netting set level one can therefore not predict with certainty the effect of this simplification. Scenarios are possible where the replacement of a Bermudan by a European option in the exposure calculation leads to a lower CVA. Hence, one cannot conclude that the simplification is a conservative practice.

\begin{figure}[H]
\begin{minipage}{.5\linewidth}
\centering
\subfloat[]{\includegraphics[scale=.48]{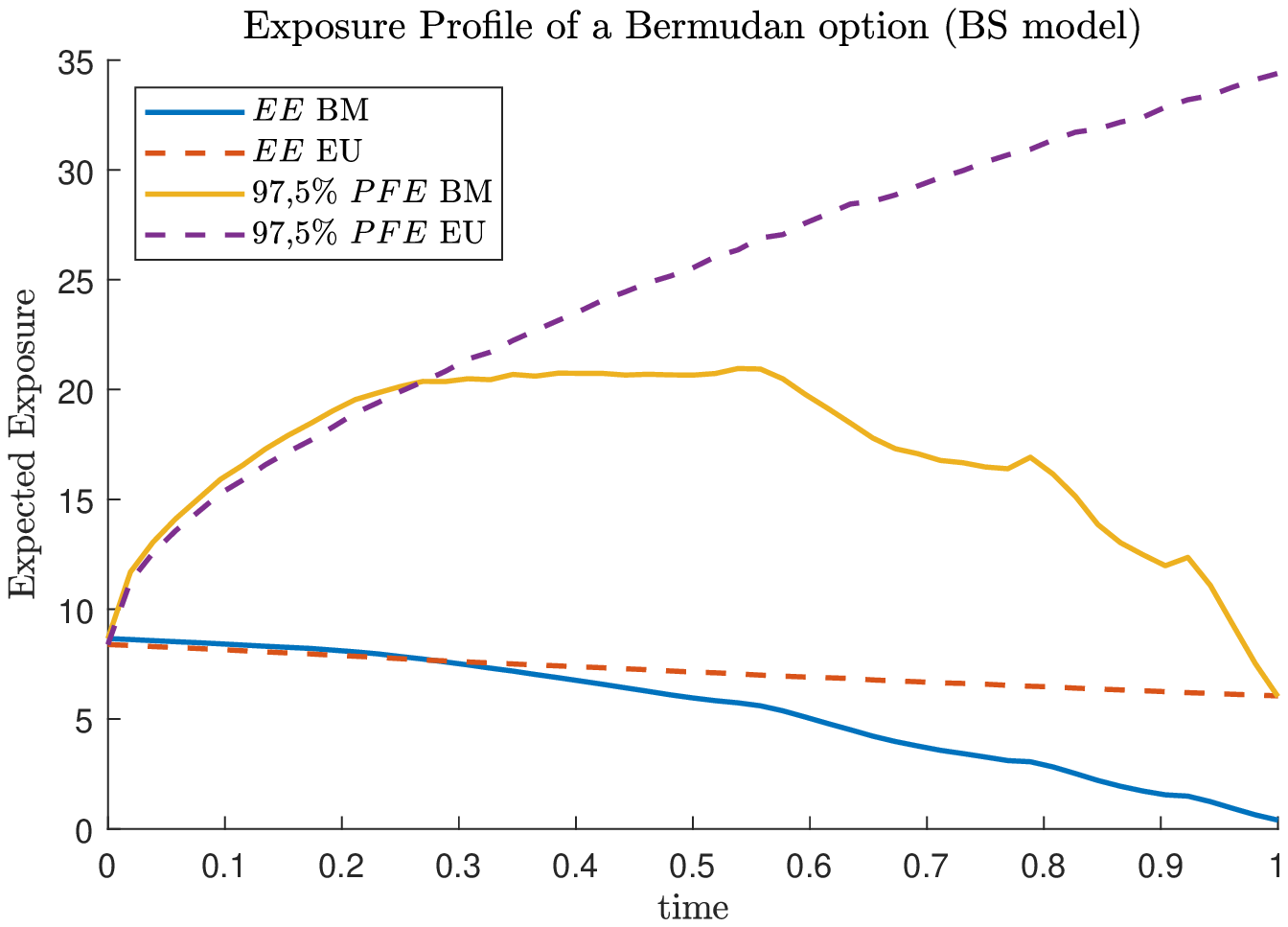}}
\end{minipage}%
\begin{minipage}{.5\linewidth}
\centering
\subfloat[]{\includegraphics[scale=.48]{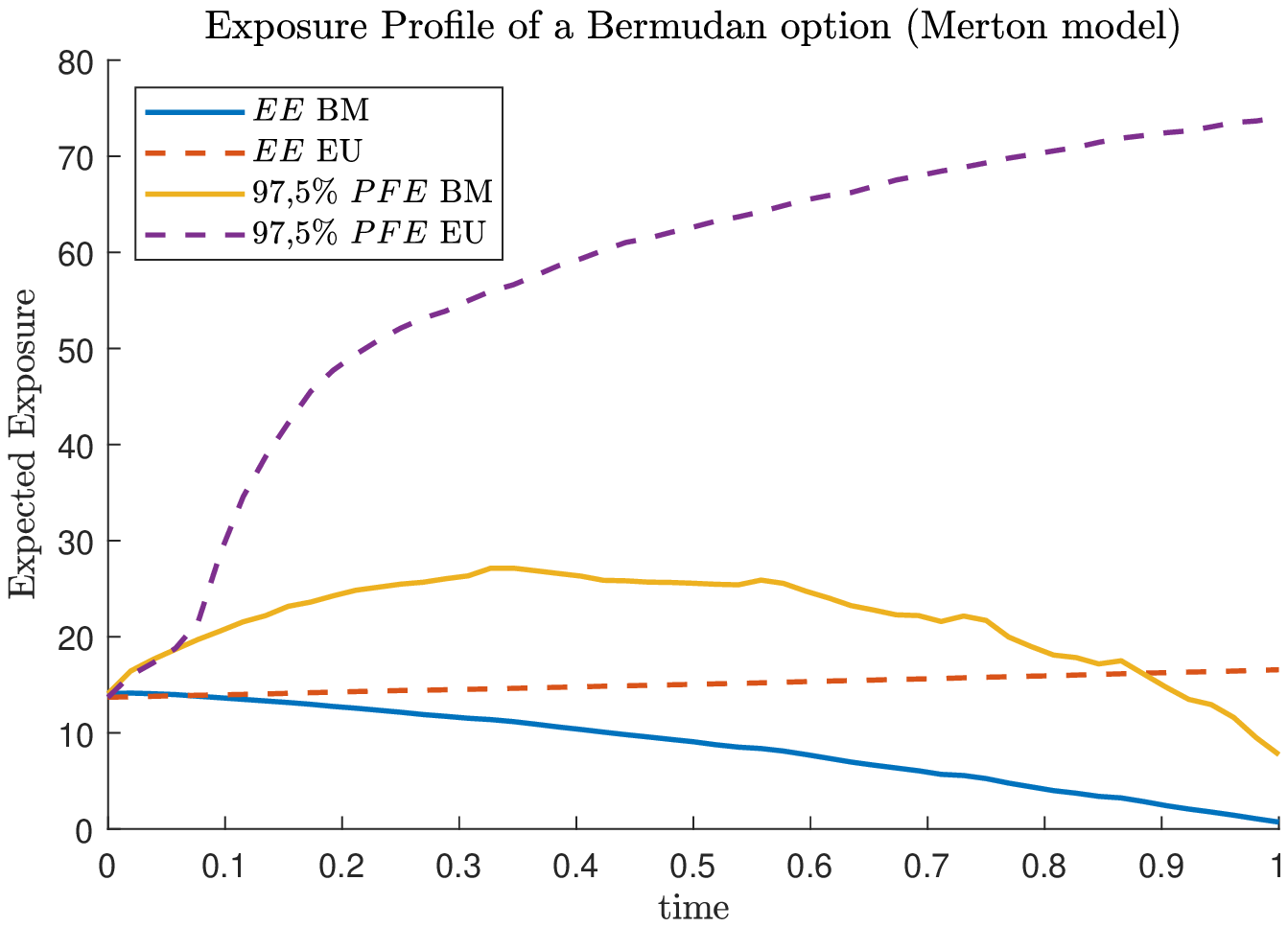}}
\end{minipage}\par\medskip
\centering
\subfloat[]{\includegraphics[scale=.48]{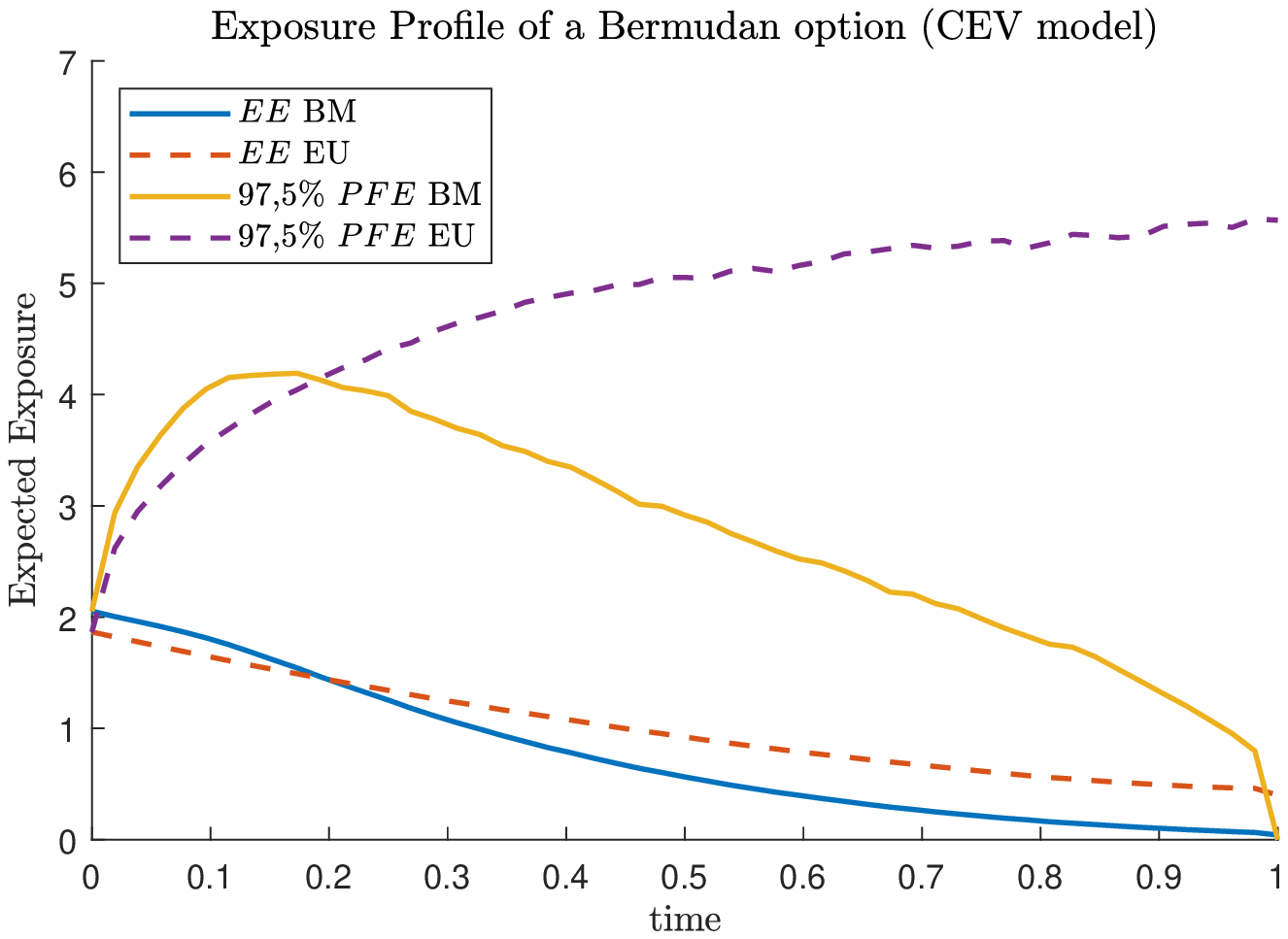}}
\caption{Expected exposure (EE) and potential future exposure (PFE) profiles for a Bermudan option and a European put option with maturity $T=1$ in the Black-Scholes model, the Merton model and the CEV model with $N_{sim}=50000$.}
\label{fig:EE_Bermudan_EU}
\end{figure}

\begin{table}[H]
\begin{center}
\begin{tabular}{llccc}
\hline 
 & & Price & EE at T & PFE at T\\ 
\hline 
\multirow{2}{*}{BS}
& European & 8.39 & 6.02 & 34.25\\ 
& Bermudan & 8.66 & 0.72 & 9.4\\
\hline 
\multirow{2}{*}{Merton}
& European & 13.69 & 16.34 & 74.07\\ 
& Bermudan & 14.07 & 1.16 & 12.62\\
\hline
\multirow{2}{*}{CEV}
& European & 2.43 & 0.75 & 8.66\\ 
& Bermudan & 2.72 & 0.07 & 0.79\\
\hline 
\end{tabular} 
\caption{Option price, expected exposure and potential future exposure of a European put and a Bermudan put option in the Black-Scholes (BS) model, the Merton model and the CEV model.}
\label{tab:EE_Bermudan_EU} 
\end{center}
\end{table}

In order to get a better understanding of the difference between European and Bermudan options we vary the number of exercise rights per year. We perform the same experiments as in the last section, however this time we consider five different Bermudan options. Figure \ref{fig:EE_BS_Bermudan_timesteps} shows the resulting exposure profile and Figure \ref{fig:PFE_BS_Bermudan_timesteps} the corresponding potential future exposure. Both are calculated on a daily basis, i.e.\;on $252$ (trading) days. We observe the exposure of a Bermudan option drops on the exercise days and between the exercise days behaves similar to a European option. The drops are smaller on the short end and become larger close to maturity. The potential future exposure shows a very similar behaviour. More exercise rights yield a smoother exposure profile for Bermudan options. Furthermore, we observe that the difference in the exposure profiles between a European option and a Bermudan option with $4$ exercise dates only is already substantial. On the other hand, the profiles for a Bermudan option with $36$ and one with $252$ exercise rights are relatively similar. The effect of adding additional exercise rights on the exposure seems to decrease with a higher number of exercise rights. 

We conclude that replacing Bermudan options by European options leads to significantly different exposure profiles. However, when it comes to efficiency, one could replace a Bermudan option with a high exercise frequency (or an American option which can always be exercised) with a Bermudan option with only a few exercise rights.
\begin{figure}[H] 
\centering
\includegraphics[width=\textwidth]{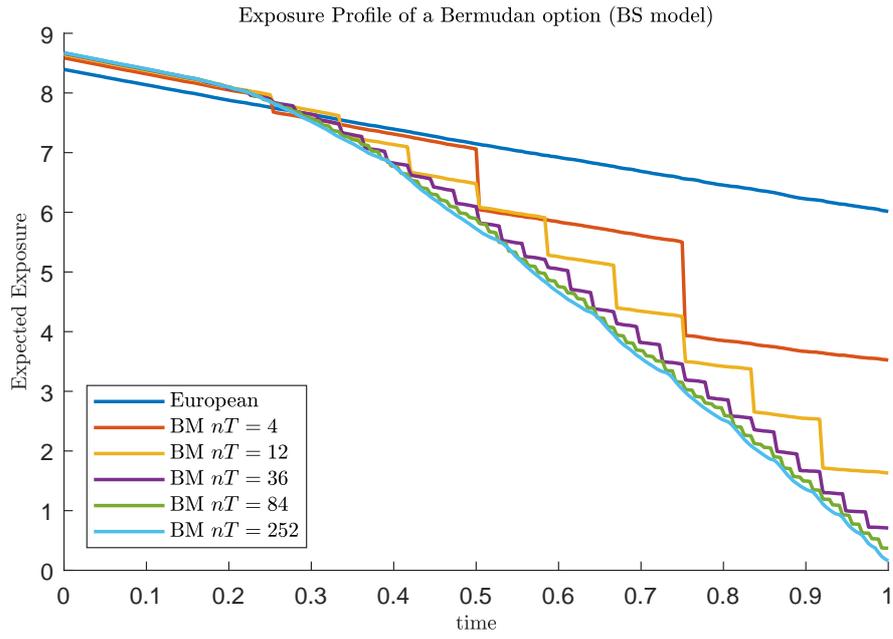} 
\caption{Expected exposure (EE) for a Bermudan option with different exercise frequencies and a European option in the Black-Scholes model with $N_{sim}=50000$.}
  \label{fig:EE_BS_Bermudan_timesteps} 
\end{figure}

\begin{figure}[H] 
\centering
\includegraphics[width=\textwidth]{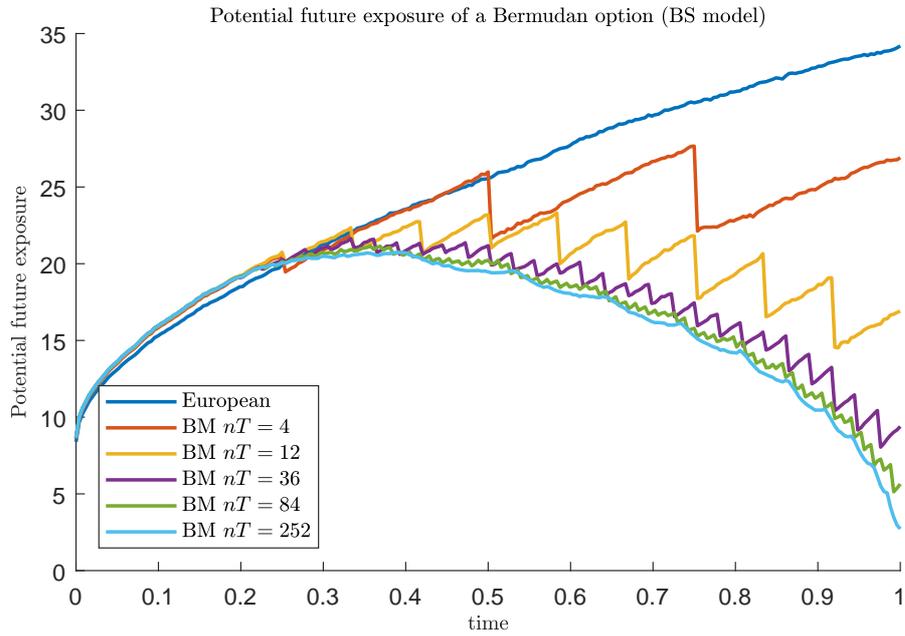} 
\caption{Potential future exposure (PFE) for a Bermudan option with different exercise frequencies and a European option in the Black-Scholes model with $N_{sim}=50000$.}
  \label{fig:PFE_BS_Bermudan_timesteps} 
\end{figure}

\begin{table}
\begin{center}
\begin{tabular}{lccc}
\hline 
& $V_{0}$ & EE at T & PFE at T \\ 
\hline 
European & 8.39 & 6.01 & 34.20 \\ 
BM $nT=4$ & 8,59 & 3.52 & 26.91 \\
BM $nT=12$ & 8.64 & 1.62 & 16.91 \\ 
BM $nT=36$ & 8.66 & 0.71 & 9.39 \\ 
BM $nT=84$ & 8.67 & 0.37 & 5.65 \\ 
BM $nT=252$ & 8.67 & 0.16 & 2.72 \\ 
\hline 
\end{tabular}
\caption{Option price, expected exposure and potential future exposure of a European and Bermudan options with different exercise frequencies in the Black-Scholes model.}
\label{tab:EE_BS_Bermudan_timesteps} 
\end{center}
\end{table}

\section{Conclusion and Outlook}
In this paper we have introduced a unified framework for the pricing and exposure calculation of European, Bermudan and barrier options based on the dynamic Chebyshev method of \cite{GlauMahlstedtPoetz2019}. The numerical experiments in Section 5 and Section 6 showed that the method is well-suited for the exposure calculation and the structure of the approach yields high efficiency. The Dynamic Chebyshev method admits several qualitative advantages.
\begin{itemize}
\item The method offers a high flexibility, the price and the credit exposure of many different products can be calculated with this algorithm and it can be used in different stock price models.
\item The structure of the method allows us to explore additional knowledge of the model by choosing different techniques to compute the conditional moments in the pre-computation. This increases the efficiency of the method compared to standard approaches such as Least Square Monte Carlo.
\item The calculated credit exposure can be aggregated on different levels and enables the efficient computation of CVA and other risk metrics on a portfolio level.
\item The algorithm is presented for barrier and Bermudan options in three different equity models. However, the approach is more general in term of models and products. One can exploit the method in interest rate, FX or commodity markets and price options such as Bermudan swaptions, callable bonds for instance.
\item In this paper, the payoff $g$ is of a standard call or put type. However, the method can easily handle more complicated payoffs as well. An example is a callable bond, where the underlying risk factor is the interest rate and the payoff is a call option on the (vanilla) bond.
\item The polynomial structure of the approximation of the value function enables an efficient computation of the option's sensitivities Delta and Gamma in every time step.
\end{itemize}
Compared to the Least Square Monte Carlo approach for CVA from \cite{Schoeftner2008} the proposed approach has quantitative and qualitative advantageous. First, \cite{GlauMahlstedtPoetz2019} show that the offline-online decomposition of the Dynamic Chebyshev method leads to an efficiency gain compared to the Least Square Monte Carlo approach of \cite{LongstaffSchwartz2001}. This is especially the case, when several options on the same underlying are priced. Moreover, the Monte Carlo approach of \cite{Schoeftner2008} requires a measure change from the pricing measure to the real-world measure. In contrast, the dynamic Chebyshev method for exposure calculation separates the pricing from the exposure calculation. Therefore one can increase the accuracy of the option price without changing the number of simulations for the exposure calculation. 

An approach which has a similar structure as the new method is presented in \cite{ShenWeideAnderluh2013}. Their approach is based on the COS method and requires the existence of the characteristic function in closed form. This means, in contrast to the Dynamic Chebyshev method, it can only be applied to a smaller class of asset models.

Furthermore, the empirical investigation of the exposure profile provides insight into the behaviour of the profiles for different option types and asset models with practical implications. In order to speed-up the exposure calculation common simplifications in practice include the choice of a simple model for the underlying risk factor and the replacement of complex options by simpler options. Our experiments reveal that the first simplification strongly affects the results for barrier option. Moreover, the experiments show that the replacement of Bermudan options by European options yields significantly different exposure profiles and two different problems occur. First, the exposure of the European option overestimates the exposure of the Bermudan for most of the option's lifetime and second we cannot conclude that this simplification is conservative. Therefore, we recommend to compute the exposure for Bermudan options directly. The presented dynamic Chebyshev method is able to do so in an efficient way.

The dynamic Chebyshev method is presented in \cite{GlauMahlstedtPoetz2019} as a general algorithm in $d$ dimensions. In this paper we focussed on the exposure calculation for products which depend only on one main risk factor. As a next step we extend the presented approach for the exposure calculation to options which have more than one main risk factor.

\appendix
\section{Proof of Proposition 4.2}
\begin{proof}
We define $\mu_{j}:=\EE[T_{j}(Y)\1_{[-1,1]}(Y)]$ as the generalized moments and $\mu^{\prime}_{j}=\EE[T^{\prime}_{j}(Y)\1_{[-1,1]}]$ as the expectations of the derivatives of the Chebyshev polynomials. The first three Chebyshev polynomials are given by $T_{0}(x)=1$, $T_{1}(x)=x$ and $T_{2}(x)=2x^{2}-1$ with derivatives $T_{0}^{\prime}(x)=0$, $T_{1}^{\prime}(x)=1=T_{0}(x)$ and $T_{2}^{\prime}(x)=4x=4T_{1}(x)$. This yields
\begin{align*}
\mu_{0}=\EE[1_{[-1,1]}(Y)]&=P(-1\leq Y\leq 1)=F(1)-F(-1).
\end{align*}
Before we consider the first moment we need the following property of the density $f$ of the normal distribution,
\begin{align*}
f^{\prime}(x)&=\frac{1}{\sqrt{2\pi}\sigma}e^{-\frac{(x-mu)^{2}}{2\sigma^{2}}}(-2\frac{(x-\mu)}{2\sigma^{2}})=f(x)(-2\frac{(x-\mu)}{2\sigma^{2}})=(-\frac{1}{\sigma^{2}})xf(x)+\frac{\mu}{\sigma^{2}}f(x),\\
&\text{and hence}\quad xf(x)=\mu f(x) - \sigma^{2}f^{\prime}(x).
\end{align*}
Using this property we obtain for the first moment $\mu_{1}=\EE[Y1_{[-1,1]}(Y)]$
\begin{align*}
\mu_{1}=\int_{-1}^{1}yf(y)\dy=\mu\int_{-1}^{1}f(y)\dy-\sigma^{2}\int_{-1}^{1}f^{\prime}(y)\dy=\mu\mu_{0}-\sigma^{2}(f(1)-f(-1)).
\end{align*}
Assume we know $\mu_{j},\mu^{\prime}_{j}$, $j=0,\ldots,n$. The Chebyshev polynomials and their derivative are recursively given by
\begin{align*}
T_{n+1}(x)=2xT_{n}(x)-T_{n-1}(x)\qquad T^{\prime}_{n+1}(x)=2(n+1)T_{n}(x)+\frac{n+1}{n-1}T^{\prime}_{n-1}(x).
\end{align*}
From the latter easily follows that
\begin{align*}
\mu^{\prime}_{n+1}&=\EE[T^{\prime}_{n+1}(Y)\1_{[-1,1]}(Y)]\\
&=2(n+1)\EE[T_{n}(Y)\1_{[-1,1]}(Y)]+\frac{n+1}{n-1}\EE[T^{\prime}_{n-1}(Y)\1_{[-1,1]}(Y)]\\
&=2(n+1)\mu_{n}+\frac{(n+1)}{(n-1)}\mu^{\prime}_{n-1}
\end{align*}
for $n\geq 2$. For the generalized moments we obtain
\begin{align*}
\mu_{n+1}=\EE[T_{n+1}(Y)\1_{[-1,1]}(Y)]=2\EE[Y T_{n}(Y)\1_{[-1,1]}(Y)]-\EE[T_{n-1}(Y)\1_{[-1,1]}(Y)].
\end{align*}
The second term is simply $\mu_{n-1}$ and for the first term we obtain
\begin{align*}
\EE[Y T_{n}\1_{[-1,1]}(Y)]&=\int_{-1}^{1}yT_{n}(y)f(y)\dy\\
&=\mu\int_{-1}^{1}T_{n}(y)f(y)\dy - \sigma^{2}\int_{-1}^{1}T_{n}(y)f^{\prime}(y)\dy\\
&=\mu\mu_{n}-\sigma^{2}\big(T_{n}(1)f(1)-T_{n}(-1)f(-1)-\mu^{\prime}_{n-1}\big).
\end{align*}
Altogether we obtain
\begin{align*}
\mu_{n+1}&=2\EE[Y T_{n}(Y)\1_{[-1,1]}(Y)]-\EE[T_{n-1}(Y)\1_{[-1,1]}(Y)]\\
&=2\big(\mu\mu_{n}-\sigma^{2}\big(T_{n}(1)f(1)-T_{n}(-1)f(-1)-\mu^{\prime}_{n-1}\big)\big)-\mu_{n-1}.
\end{align*}
It remains to find an expression for $\mu^{\prime}_{n}$.\\

We will prove by induction that
\begin{align}\label{eq:Cheby_deriv_recursive}
\mu^{\prime}_{n+1}=2(n+1)\sum_{j=0}^{n}{}^{'}\mu_{j}\1_{(n+j)\bmod 2=0}, \quad n\geq 0
\end{align}
where $\sum{}^{'}$ indicates that the first term is multiplied with $1/2$. For $n=0$, we obtain
\begin{align*}
\mu^{\prime}_{1}=2\sum_{j=0}^{0}{}^{'}\mu_{j}\1_{(0+j)\bmod 2=0}=2\frac{1}{2}\mu_{0}=\1_{0\bmod 2=0}=\mu_{0}.
\end{align*}
which shows \eqref{eq:Cheby_deriv_recursive}. Assume \eqref{eq:Cheby_deriv_recursive} holds for $j=0,\ldots,n$. Then we obtain
\begin{align*}
\mu^{\prime}_{n+1}&=2(n+1)\mu_{n}+\frac{(n+1)}{(n-1)}\mu^{\prime}_{n-1}\\
&=2(n+1)\mu_{n}+\frac{(n+1)}{(n-1)}2(n-1)\sum_{j=0}^{n-2}{}^{'}\mu_{j}\1_{(n-2+j)\bmod 2=0}\\
&=2(n+1)\Big(\mu_{n}\1_{(n+n)\bmod 2=0}+\mu_{n-1}\1_{(n+n-1)\bmod 2=0}+\sum_{j=0}^{n-2}{}^{'}\mu_{j}\1_{(n+j)\bmod 2=0}\Big)\\
&=2(n+1)\sum_{j=0}^{n}{}^{'}\mu_{j}\1_{(n+j)\bmod 2=0}.
\end{align*}
We use that $(n+j)\bmod 2=(2+j-2)\bmod 2$. For the generalized moments we thus obtain
\begin{align*}
\mu_{n+1}=2\mu\mu_{n} - 2\sigma^{2}\big(f(1)-f(-1)T_{n}(-1)-2(n-1)\sum_{j=0}^{n-2}{}^{'}\mu_{j}\1_{(n+j)\bmod 2=0}\big)-\mu_{n-1}
\end{align*}
which was our claim.
\end{proof}
\bibliographystyle{chicago}
  \bibliography{CVA_Literature}

\end{document}